\documentclass[runningheads]{llncs}
\usepackage[T1]{fontenc}
\usepackage{graphicx}

\usepackage{amssymb}
\usepackage{graphicx}
\usepackage{subcaption}
\usepackage{hyperref}
\usepackage{amsmath}
\usepackage{comment}
\usepackage{breqn}
\usepackage[ruled, linesnumbered]{algorithm2e}
\usepackage{booktabs, longtable, multicol, adjustbox}
\usepackage{tabularx}
\usepackage{cleveref}
\usepackage{todonotes}

\usepackage{geometry}
\usepackage{soul, xcolor}	

\newcommand{\er}[1]{\textcolor{blue}{#1}}
\newcommand{\erel}[1]{\er{(Erel says: #1)}}

\newcommand{\kbp}{$k\text{BP}$}

\begin{document}
\title{\texorpdfstring{$k$}{k}-times bin packing and its application to fair electricity distribution \thanks{This is an extended version of our paper accepted at SAGT-2024 (\url{https://link.springer.com/chapter/10.1007/978-3-031-71033-9_27}. This arXiv version corrects a bug present in the accepted SAGT paper that affected section 7.3.}}


%
%
\author{Dinesh Kumar Baghel\inst{1}
	\and
	Alex Ravsky\inst{2}
	\and
	Erel Segal-Halevi\inst{1}
}
%
%
\institute{Ariel University, Ariel 40700, Israel \\
	\email{\{dinkubag21, erelsgl\}@gmail.com}\\
	\and
	Pidstryhach Institute for Applied Problems of Mechanics and Mathematics of National Academy of Sciences of Ukraine, Lviv, Ukraine\\
	\email{alexander.ravsky@uni-wuerzburg.de}}
\maketitle             

\begin{abstract}
	 Given items of different sizes and a fixed bin capacity, the bin-packing problem is to pack these items into a minimum number of bins such that the sum of item
	sizes in a bin does not exceed the capacity. 
	We define a new variant called \emph{$k$-times bin-packing ($k$BP)}, where the goal is to pack the items such that each item appears exactly $k$ times, in $k$ different bins. 
	We generalize some existing approximation algorithms for bin-packing to solve $k$BP, and analyze their performance ratio.
	
	The study of $k$BP is motivated by the problem of \emph{fair electricity distribution}. In many developing countries, the total electricity demand is higher than the supply capacity. 
	We prove that every electricity division problem can be solved by $k$-times bin-packing for some finite $k$.
	We also show that $k$-times bin-packing can be used to distribute the electricity in a fair and efficient way.
	Particularly, we implement generalizations of the First-Fit and First-Fit Decreasing  bin-packing algorithms to solve $k$BP, and apply the generalizations to real electricity demand data. We show that our generalizations outperform existing heuristic solutions to the same problem in terms of the egalitarian allocation of connection time.  

	
\end{abstract}

\section{Introduction} \label{sec:introduction}
This work is motivated by the problem of \emph{fair electricity distribution}. 
In developing countries, the demand for electricity often surpasses the available supply \cite{Kaygusuz2012}.  
Such countries have to come up with a fair and efficient method of allocating of the available electricity among the households.

Formally, we consider a power-station that produces a fixed supply $S$ of electricity. 
The station should provide electricity to $n$ households. The demands of the households in a given period are given by a (multi)set $D$. Typically, $\sum_i D[i] > S$ (where $D[i]$ is the electricity demand of a household $i$), 
so it is not possible to connect all households simultaneously.
Our goal is to ensure that each household is connected the same amount of time, and that this amount is as large as possible.
We assume that an agent gains utility only if the requested demand is fulfilled; otherwise it is zero. Practically it can be understood as follows: Suppose at some time, a household $i$ is running some activity that requires $D[i]$ kilowatt of electricity to operate; in the absence of that amount, the activity will not function. 
Therefore, an allocation where demands are fractionally fulfilled is not relevant. 

A simple approach to this problem is to partition the households into some $q$ subsets, such that the sum of demands in each subset is at most $S$, and then connect the agents in each subset for a fraction $1/q$ of the time. To maximize the amount of time each agent is connected, we have to minimize $q$. This problem is equivalent to the classic problem of \emph{bin-packing}. In this problem, we are given some $n$ items, of sizes given by a multiset of positive numbers numbers $D$, and a positive number $S$ representing the capacity of a bin. The goal  is to pack items in $D$ into the smallest possible number of bins, such that the sum of item sizes in each bin is at most $S$. The problem is NP-complete \cite{10.5555/574848}, but has many efficient approximation algorithms.

However, even an optimal solution to the bin-packing problem may provide a sub-optimal solution to the electricity division problem. As an example, suppose we have three households $x,y,z$ with demands $2,1,1$ respectively, and the electricity supply $S = 3$. Then, the optimal bin-packing results in $2$ bins, for instance, $\{x,y\}$ and $\{z\}$. This means that each agent would be connected $1/2$ of the time. However, it is possible to connect each agents $2/3$ of the time, by connecting each of the pairs $\{x,y\}, \{x,z\}, \{y,z\}$ for $1/3$ of the time,
as each agent appears in $2$ different subsets.
More generally, suppose we construct $q$ subsets of agents, such that each agent appears in exactly $k$ different subsets. Then we can connect each subset for $1/q$ of the time, and each agent will be connected $k/q$ of the time.

\subsection{The \texorpdfstring{$k$}{k}-times bin-packing problem} 
To study this problem more abstractly, we define the \emph{$k$-times bin-packing problem (or \kbp)}. 
The input to \kbp{} is a set of $n$ items of sizes given by a multiset $D$,
a positive number $S$ representing the capacity of a bin, and an integer $k\geq 1$.
The goal  is to pack items in $D$ into the smallest possible number of bins, such that the sum of item sizes in each bin is at most $S$, and each item appears in  $k$ different bins, where each item occurs at most once in a bin.
In the above example, $k=2$. 
It is easy to see that, in the above example, $2$-times bin-packing yields the optimal solution to the electricity division problem.

Our first main contribution \textbf{(Section \ref{section:existence of finite k for kBP})} is to prove that, for every electricity division problem, there exists some finite $k$ for which the optimal solution to the \kbp{} problem yields the optimal solution to the electricity division problem.

We note that \kbp{} may have other applications beyond electricity division. For example, it could be to create a backup of the files on different file servers \cite{Jansen_1998}.
We would like to store $k$ different copies of each file, but obviously, we want at most one copy of the same file on the same server. This can be solved by solving \kbp{} on the files as items, and the server disk space as the bin capacity.

Motivated by these applications, we would like to find ways to efficiently solve \kbp{}. However, it is well-known that \kbp{} is NP-hard even for $k=1$. We therefore look for efficient approximation algorithms of \kbp{}. 

\subsection{Using existing bin-packing algorithms for \texorpdfstring{$k$}{k}BP} \label{using existing bin packing algorithms for kBP}
Several existing algorithms for bin-packing can be naturally extended to $k$BP. However, it is not clear whether the extension will have a good approximation ratio.

As an example, consider the simple algorithm called \emph{First-Fit (FF)}: process the items in an arbitrary order; pack each item into the first bin it fits into; if it does not fit into any existing bin, open a new bin for it. 
In the example $D = [10,20,11], S=31$, the $FF$ would pack two bins: $\{10,20\}$ and $\{11\}$. This is clearly optimal. 
The extension of $FF$ to $k$BP would process the items as follows: for each item $x_r$ in the list (in order), suppose that $b$ bins have been used thus far.  Let $j$ be the lowest index ($1 \leq j\leq b$) such that (a) bin $j$ can accommodate $x_r$ and (b) bin $j$ does not contain any copy of $x_r$, should such $j$ exist; otherwise open a new bin with index $j=b+1$.  Place $x_r$ in bin $j$.

There are two ways to process the input. One way 
is by processing each item $k$ times in sequence. In the above example, with $k=2$, $FF$ will process the items in order $[10^1, 10^2, 20^1, 20^2, 11^1, 11^2]$, where the superscript specifies the instance to which an item belongs.
This results in four bins: \sloppy $\{10^1, 20^1\}, \{10^2, 20^2\}, \{11^1\}, \{11^2\}$, which simply repeats $k$ times the solution obtained from $FF$ on $D$. However, the optimal solution here is 3 bins: $\{10^1,20^1\}, \{11^1,10^2\}, \{20^2,11^2\}$.

Another way 
is to process the whole sequence $D$, $k$ times. In the above example, $FF$ will process the sequence $D_2 = DD = [10^1,20^1,11^1, 10^2, 20^2, 11^2]$ . 
Applying the $FFk$ algorithm to this input instance will result in three bins $\{10^1,20^1\}, \{11^1,10^2\}, \{20^2,11^2\}$, which is optimal.
Thus, while the extension of $FF$ to $k$BP is simple, it is not trivial, and it is vital to study the approximation ratio of such algorithms in this case. 

As another example, consider the approximation schemes by 
de la Vega and Lueker \cite{fernandez_de_la_vega_bin_1981} and Karmarkar and Karp \cite{karmarkar-efficient-1982}.
These algorithms use a linear program that counts the number of bins of each different \emph{configuration} in the packing (see subsection {\ref{section:EAA:subsection:general}} for the definitions) 
One way in which these algorithms can be extended, without modifying the linear program, is to give $D_k$  as the input. But then, a configuration might have more than one copy of an item in $D$, which violates the $k$BP constraint. 
Another approach is to modify the constraint in the configuration linear program, to check that there are $k$ copies of each item in the solution, while keeping the same configurations as for the input $D$. Doing so will respect the $k$BP constraint.
Again, while the extension of the algorithm is straightforward, it is not clear what the approximation ratio would be; this is the main task of the present paper.

The most trivial way to extend existing algorithms is  to run an existing bin-packing algorithm, and duplicate the output $k$ times. However, this will not let us enjoy the benefits of $k$BP for electricity division (in the above example, this method will yield $4$ bins, so each agent will be connected for $2/4=1/2$ of the time).
Therefore, we present more elaborate extensions, that attain better performance. 
The algorithms we extend can be classified into two classes:

\paragraph{1. Fast constant-factor approximation algorithms \textbf{(Section \ref{section:approx_algo})}} 
Examples are First-Fit ($FF$) and First-Fit-Decreasing ($FFD$). 
For bin-packing, these algorithms find a packing with at most $1.7 \cdot OPT(D)$ and  $\frac{11}{9} \cdot {OPT}(D) + \frac{6}{9}$  bins respectively \cite{Dsa2013,dosa_tight_2007,dosa_tight_2013}
We adapt these algorithms by running them on an instance made of $k$ copies, $DD\ldots D$ ($k$ copies of $D$),
which we denote by $D_k$.
We show that, for $k>1$, the extension of $FF$ to $k$BP (which we call $FFk$) finds a packing with at most  $\left(1.5+\frac{1}{5k}\right) \cdot OPT(D_k) + 3\cdot k$ bins.
For any fixed $k>1$, the asymptotic approximation ratio of $FFk$ for large instances (when $OPT(D_k)\to \infty$)
	is $(1.5+\frac{1}{5k})$, which is better than that of $FF$, and improves towards $1.5$ when $k$ increases.

We also prove that the lower bound for $FFDk$ (the extension of $FFD$ to $k$BP) is $\frac{7}{6} \cdot {OPT}(D_k) + 1$, and conjecture by showing on simulated data that $FFDk$ solves $k$BP with at most $\frac{11}{9} \cdot OPT(D_k) + \frac{6}{9}$ bins which gives us an asymptotic approximation ratio of at most $11/9$. 

We also show that the extension of $NF$ (next-fit algorithm) to $k$BP (we call this extension as $NFk$) has the asymptotic ratio of 2.

\paragraph{2. Polynomial-time approximation schemes \textbf{(Section \ref{section:ptas})}}
Examples are the algorithms by Fernandez de la Vega and Lueker \cite{fernandez_de_la_vega_bin_1981} and Karmarkar and Karp \cite{karmarkar-efficient-1982}.
We show that the algorithm by Fernandez de la Vega and Lueker can be extended to solve $k$BP using at most $(1 + 2\cdot \epsilon)OPT(D_k) + k$ bins for any fixed $\epsilon \in (0,1/2)$. 
For every $\epsilon >0$, Algorithm 1 of Karmarkar-Karp algorithms \cite{karmarkar-efficient-1982} solves $k$BP using bins at most  $ (1 + 2 \cdot k \cdot \epsilon)OPT(D_k) + \frac{1}{2 \cdot \epsilon^2} + (2 \cdot k+1)$ bins, and runs in time
$O(n(D_k) \cdot \log {n(D_k)} + T(\frac{1}{\epsilon^2}, n(D_k))$, where $n(D_k)$ is the number of items in $D_k$, and  $T$ is a polynomially-bounded function.
Algorithm 2 of Karmarkar-Karp algorithms \cite{karmarkar-efficient-1982} generalized to solve $k$BP using at most  $OPT(D_k) + O(k \cdot \log^2 {OPT(D)})$ bins, and runs in time
$O(T(\frac{n(D)}{2},n(D_k)) + n(D_k) \cdot \log{n(D_k)})$.

\paragraph{Simulation experiments \textbf{(Appendix \ref{subsection:ffk and ffdk bp algorithms experiment})}}
We complement our theoretical analysis of the approximation ratios of $FFk$ and $FFDk$ bin-packing algorithms with
simulation experiments, which validate our conjectures in subsection {\ref{section:approx-algo:subsection:approx_algo:ffk}} and {\ref{section:approx-algo:subsection:approx_algo:ffdk}}.

\paragraph{Electricity distribution \textbf{(Section \ref{section:experiments})}}
The fair electricity division problem was introduced by Oluwasuji, Malik, Zhang and Ramchurn
\cite{oluwasuji_algorithms_2018,oluwasuji_solving_2020} under the name of ``fair load-shedding''.
They presented several heuristic as well as ILP-based algorithms, and tested them on a dataset of $367$ households from Nigeria. 
We implement the $FFk$ and $FFDk$ algorithms for finding approximate solutions to $k$BP, and use the solutions to determine a fair electricity allocation. We test the performance of our allocations on the same dataset of Oluwasuji, Malik, Zhang and Ramchurn \cite{oluwasuji_algorithms_2018}.
We compare our results on the same metrics used by Oluwasuji, Malik, Zhang and Ramchurn \cite{oluwasuji_algorithms_2018}. These metrics are utilitarian and egalitarian social welfare and the maximum utility difference between agents. We compare our results in terms of hours of connection to supply on average, utility delivered to an agent on average, and electricity supplied on average, along with their standard deviation. We find that our results surmount their results in terms of the egalitarian allocation of connection time to the electricity 
$FFk$ and $FFDk$ run in time that is nearly linear in the number of agents. We conclude that using $k$BP can provide a practical, fair and efficient solution to the electricity division problem where the objective is to connect each agent as much as possible.

In Section \textbf{\ref{section:conclusion}}, we conclude with a summary and directions for future work. We defer most of the technical proofs and algorithms to the Appendix.

\section{Related Literature} \label{section:related-literature}
\paragraph{First-Fit}

We have already defined the working of $FF$ in subsection \ref{using existing bin packing algorithms for kBP}.
Denote by $FF$ the number of bins used by the First-Fit algorithm, and by $OPT$ the number of bins in an optimal solution for a multiset $D$. An upper bound of $FF \leq 1.7OPT+3$ was first proved by Ullman in 1971 \cite{laboratory_performance_1971}. The additive term was first improved to 2 by Garey, Graham and Ullman \cite{10.1145/800152.804907} in 1972. In 1976, Garey, Graham, Johnson and Yao \cite{GAREY1976257}  improved the bound further to $FF \leq \lceil 1.7OPT \rceil$, equivalent to $FF  \leq 1.7OPT + 0.9$ due to the integrality of $FF$ and $OPT$. This additive term was further lowered to $FF \leq 1.7OPT + 0.7$ by Xia and Tan \cite{xia_tighter_2010}. Finally, in 2013 Dosa and Sgall \cite{Dsa2013}
settled this open problem and proved that $FF \leq \lfloor 1.7OPT \rfloor$, which is tight.

\paragraph{First-Fit Decreasing} Algorithm \emph{First-Fit Decreasing} ($FFD$) first sorts the items in non-increasing order, and then implements $FF$ on them. 
In 1973, in his doctoral thesis \cite{Johnson1973}, D. S. Johnson proved that $ FFD\ \leq \frac{11}{9}OPT+4$. Unfortunately, his proof spanned more than 100 pages. In 1985, Baker \cite{baker_new_1985} simplified their proof and improved the additive term to $3$. In 1991 Minyi \cite{Minyi} further simplified the proof and showed that the additive term is $1$. Then, in 1997,  Li and Yue \cite{Li_Yue_1997} narrowed the additive constant to $7/9$ without formal proof. Finally, in 2007 Dosa \cite{dosa_tight_2007} proved that the additive constant is $6/9$. They also gave an example which achieves this bound. 

\paragraph{Next-fit} The algorithm next-fit works as follows: It keeps the current bin (initially empty) to pack the current item. If the current item does not pack into the currently open bin then it closes the current bin and opens a new bin to pack the current item. Johnson in his doctoral thesis \cite{Johnson1973} proved that the asymptotic performance ratio of next-fit is 2.

\paragraph{Efficient approximation schemes} In 1981, Fernandez de la Vega and Lueker \cite{fernandez_de_la_vega_bin_1981} presented a polynomial time approximation scheme to solve bin-packing. Their algorithm accepts as input an $\epsilon> 0$ and produces a packing of the items in $D$ of size at most $\left(1+\epsilon\right)OPT+1$. Their running time is polynomial in the size of $D$ and depends on $1/\epsilon$. They invented the \emph{adaptive rounding} method to reduce the problem size. In adaptive rounding, they initially organize the items into groups and then round them up to the maximum value in the group. This results in a problem with a small number of different item sizes, which can be solved optimally using the linear configuration program. Later, Karmarkar and Karp \cite{karmarkar-efficient-1982} devised several PTAS for the bin-packing problem. One of the Karmarkar-Karp algorithms solves bin-packing using at most $OPT + O({\log^2\,{ OPT}})$ bins. Other Karmarkar--Karp algorithms have different additive approximation guarantees, and they all run in polynomial time. This additive approximation was further improved to $O(\log\,{OPT} \cdot \log\,{\log\,{OPT}})$ by Rothvoss \cite{Rothvoss_2013}. They used a ``glueing" technique wherein they glued small items to get a single big item. In 2017, Hoberg and Rothvoss \cite{Hoberg_Rothvoss_2017}  further improved the additive approximation to a logarithmic term $O(\log\,{OPT})$.

Jansen \cite{Jansen_1998} has proposed a $\mathrm{FPTAS}$ for the generalization of the bin-packing problem called as \textit{bin-packing with conflicts}. 
The input instance for their algorithm is the conflict graph. Its vertices are the items and any two items are adjacent provided they cannot be packed into the same bin. In particular, $k$BP can be considered as the bin-packing with conflicts, where the conflict graph $D_k$ is a disjoint union of copies of a complete graph $K_k$. Their bin-packing problem with conflicts is restricted to $q-$inductive graphs. In a $q-$inductive graph the vertices are ordered from $1, \ldots, n$. Each vertex in the graph has at most $q$ adjacent lower numbered vertices. 
Since the degree of each vertex of $D_k$ equals $k-1$, $D_k$ is a $k-1$-inductive graph. 
In their method first they obtain an instance of large items from the given input instance. Let this instance be $J_k$. They apply the linear grouping method of Fernandez de la Vega and Lueker \cite{fernandez_de_la_vega_bin_1981} to obtain a constant number of different item sizes. Next they apply the Karmarkar and Karp algorithm \cite{karmarkar-efficient-1982} to obtain an approximate packing of the large items. The bins in this approximate packing may have conflicts, so they use the procedure called $\mathrm{COLOR}$ which places each conflicted item into a new bin. In the worst case it may happen that all the items in each bin have conflict and hence each one of them is packed into a separate bin. Finally, after removing the conflicts, they packed the small items into the existing bins, respecting conflicts among items. In doing so, new bins are opened if necessary. 
In this paper we focus on a special kind of conflicts, and for this special case, we present a better approximation ratio.
Their algorithm solves the $k$BP using at most $(1 + 2 \cdot \epsilon) OPT(D_k) + \frac{2 \cdot k -1 }{4 \cdot \epsilon^2} + 3\cdot k + 1$ bins, 
whereas our extension to Algorithm 1 and 2 of Karmarkar-Karp algorithms solves the $k$BP using at most $(1 + 2 \cdot \epsilon) OPT(D_k) + \frac{1 }{2 \cdot \epsilon^2} + 2\cdot k + 1$ and $OPT(D_k) + O(k \cdot \log^2 OPT(D))$ bins respectively.

Gendreau, Laporte and Semet \cite{Gendreau_Laporte_Semet_2004} propose six heuristics named H1 to H6 for bin-packing with item-conflicts, represented by a general conflict graph.
The heuristic H1 is a variant of $FFD$ which  incorporates the conflicts, whereas H6 is a combination of a maximum-clique procedure and $FFD$. They show that H6 is better than H1 for conflict graphs with high density, whereas H1 performs marginally better for low density conflict graphs (where density is defined as the ratio of the number of edges to the number of possible edges).
\kbp{} can be represented by duplicating each item $k$ times, and constructing a conflict graph in which there are edges between each two copies of the same item. The density of this graph is $(k-1)/(kn-1)$, which becomes smaller for large $n$. This suggests that H1 is a better fit for \kbp{}. But in their adoption of \kbp{} the vertices of the conflict graphs are ordered as blocks corresponding to the items, in such a way that the their sizes are non-decreasing. We have already seen that when we change such item order we can obtain a better packing.

Recently, Doron-Arad, Kulik and Shachnai \cite{Doron-Arad_Kulik_Shachnai_2022} has solved in polynomial time a more general variant of bin-packing, with partition matroid constraints. Their algorithm packs the items in $OPT + O\left(\frac{OPT}{(\ln {\ln OPT})^{1/17}}\right)$ bins. Their algorithm can be used to solve the $k$BP:
for each item in $D$, define a category that contains $k$ items with the same size.
Then, solve the bin-packing with the constraint that each bin can contain at most one item from each category (it is a special case of a partition-matroid constraint).
However, in the present paper we focus on the special case of $k$BP. This allows us to attain a better running-time (with $FFk$ and $FFDk$), and a better approximation ratio (with the de la Vega--Lueker and Karmarkar--Karp algorithms).

\paragraph{Cake cutting and electricity division} The electricity division problem can also be modeled as a classic resource allocation problem known as `cake cutting'. The problem was first proposed by Steinhaus {\cite{Society2016}}. A number of cake-cutting protocols have been discussed in {\cite{brams_taylor_1996,Webb_1998}}. In cake cutting, a cake is a metaphor for the resource. Like previous approaches, a time interval can be treated as a resource. A cake-cutting protocol then allocates this divisible resource among agents who have different valuation functions (or preferences) according to some fairness criteria. The solution to this problem differs from the classic cake-cutting problem in the sense that at any point in time, $t$, the sum of the demands of all the agents, should respect the supply constraint, and several agents may share the same piece.

Load-shedding strategy at the bus level has been discussed in \cite{Shi2015}. They have dealt with the issue of cascading failure, which may occur in line contingencies. So when a line contingency happens, first they do load-shedding to avoid cascading failure. This load-shedding may reduce the loads beyond what is necessary (let's call this unnecessary load) to prevent cascading failure. Therefore, they recover as much unnecessary load as possible in the second step while maintaining the system's stability.
Shi and Liu \cite{Shi2015} presented a distributed algorithm for compensating agents for load-shedding based on the proportional-fairness criterion. In contrast, we present a centralized algorithm for computing load-shedding that attains a high egalitarian welfare.

\section{Definitions and Notation} 
\label{section:model-assumptions}
\subsection{Electricity division problem}
The input to Electricity Division consists of:
\begin{itemize}
    \item A number $S>0$ denoting the total amount of available supply (e.g. in kW);
    \item A number $n$ of households, and a list $D = D[1],\ldots,D[n]$ of positive numbers, where $D[i]$ represents the demand of households $i$ (in kW);
    \item An interval $[0,T]$ representing the time in which electricity should be supplied to the households.
\end{itemize}
 
 The desired output consists of:
 \begin{itemize}
     \item A partition $\mathcal{I}$ of the interval $[0,T]$ into sub-intervals, $I_1,\ldots,I_p$;
     \item For each interval $l\in [p]$, a set $A_l \subseteq [n]$ denoting the set of agents that are connected to electricity during interval $l$, such that
    $\sum_{i \in A_l} D[i] \leq S$ (the total demand is at most the total supply).
 \end{itemize}
Throughout most of the paper, we  assume that the utility of agent $i$ equals the total time  agent $i$ is connected:
$u_i(\mathcal{I}) = \sum_{l: i\in A_l}|I_l|$
(we will consider other utility functions at \Cref{section:experiments}).

The optimization objective is 
\begin{align*}
\max_{\mathcal{I}} \min_{i\in [n]} u_i(\mathcal{I}),   
\end{align*}
where the maximum is  over all partitions that satisfy the demand constraints. 
This max-min value is called the \emph{egalitarian connection-time} in the given instance.


\subsection{$k$-times bin packing}
We denote the bin capacity 
by $S>0$ and the multiset of $n$ items by $D$. 

Let $n(D)$ and $m(D)$ denote the number of items and the number of different item sizes in $D$, respectively. We denote these sizes by $c[1],\dots, c[m(D)]$. Moreover, for each natural $i\le m(D)$ let $n[i]$ be the number of items of size $c[i]$. The \emph{size} of a bin is defined as the sum of all the item sizes in that bin. Given a multiset $B$ of items, we assume that its \emph{size} $V(B)$ equals the sum of the sizes of all items of $B$.

We denote $k$ copies  of $D$ by $D_k := DD\ldots D$. We denote the number of bins used to pack the items in $D_k$ by the optimal and the considered algorithm by $OPT(D_k)$ and $bins(D_k)$, respectively. 

Note that each item in $D_k$ is present at most once in each bin, so it is present in exactly $k$ distinct bins. Consider the example in Section $\ref{sec:introduction}$. There are three items $x,y,z$ with demand $2,1,1$ respectively. Let $k=2$ and $S=3$. 
Then, $\{x,y\}, \{y,z\}, \{z,x\}$ is a valid bin-packing. Note that each item is present twice overall, but at most once in each bin.
In contrast, the bin-packing $\{x,y\}, \{y,z,z\}, \{x\}$ is not valid, because there are two copies of $z$ in the same bin.

\section{On optimal \texorpdfstring{$k$}{k} for \texorpdfstring{$k$}{k}-times bin-packing} \label{section:existence of finite k for kBP}
In this section we prove that, for every electricity division instance, there exists an  integer $k$ such that $k$BP yields the optimal electricity division. Moreover, we give an upper bound on $k$ as a function of the number of agents.

\subsection{Upper bound}
Let $X$ be a nonempty set. 
We denote by $\mathbb{R}^X$ the linear space of all functions from $X$ to the real numbers $\mathbb{R}$. So elements of $\mathbb{R}^X$ have the form $(w_\alpha)_{\alpha \in X}$, where for each element $\alpha\in X$, $w_\alpha$ is the corresponding real number.

For each nonempty subset $Y$ of $X$, let $\pi_Y : \mathbb{R}^X \rightarrow \mathbb{R}^Y$ be the natural projection, which maps each element $(w_\alpha)_{\alpha \in X} \in \mathbb{R}^X$ to the element 
$(w_\alpha)_{\alpha \in Y} \in \mathbb{R}^Y$.

We shall need the following lemmas.
\begin{lemma}
	\label{lem:Y}
	Let $W\subseteq \mathbb{R}^X$ be a nonempty finite linearly independent set. Then there exists a subset $Y$ of $X$ with $|Y|=|W|$ such that the set $\pi_Y(W)$ is linearly independent.
\end{lemma}
For example, let the set of indices $X$ be $\{1,2,3,4,5\}$, so $\mathbb{R}^X$ is simply $\mathbb{R}^5$.
Let $W = \{ [1,0,0,0,0], [0,0,1,0,0], [0,0,0,0,1] \}$. As $|W|=3$, the lemma says that there exists a subset $Y$ containing at most $3$ indices, such that the projection of $W$ on these indices is still linearly independent. Indeed, in this case we can take $Y = \{1,3,5\}$.
\begin{proof}[of \Cref{lem:Y}]
	Let $q=|W|$. We prove the required claim by induction on $q$. The base case is $q = 1$. 
	Then $W$ consists of a single nonzero vector $w$.
	Therefore there exists $\alpha\in X$ such that the $\alpha$th entry of $w$ is non-zero. Put $Y=\{\alpha\}$. Then the vector $\pi_Y(w)$ is non-zero and so the set $\{\pi_Y(w)\}$ is linearly independent. 
	
	Now suppose that the required claim holds for $q-1$. Pick any vector $w'\in W$ and put $W'=W\setminus\{w'\}$. 
	By the induction assumption,
	there exists a subset $Y'$ of $X$ with $|Y'|=q-1$ such that the set $\pi_{Y'}(W')$ is linearly independent. 
	As $|W'|=q-1$, adding a single vector $\pi_{Y'}(w')$ to $\pi_{Y'} (W')$ makes it linearly dependent.
	Therefore there exists a unique vector of coefficients $(\lambda_v)_{v\in W'}$ such that $\pi_{Y'}(w')=\sum_{v\in W'} \lambda_v \pi_{Y'} (v)$. Since the set $W$ is linearly independent, 
	$\sum_{v\in W'} \lambda_v v \neq w'$,
	so there exists $\alpha\in X\setminus Y'$ such that $w'_\alpha\ne\sum_{v\in W'} \lambda_v v_\alpha$. Put $Y=Y'\cup\{\alpha\}$. 
	
	We claim that the set $\pi_Y(W)$ is linearly independent. Indeed, suppose for a contradiction that there exist coefficients $(\lambda'_v)_{v\in W}$ which are not all zeroes such that $\sum_{v\in W} \lambda'_v\pi_Y(v)=0$. Since the set $\pi_{Y'}(W')$ is linearly independent, the set $\pi_{Y}(W')$ is linearly independent too,
	so $\lambda'_{w'}\ne 0$. Then $\pi_Y(w')=\sum _{v\in W'} (-\lambda'_v/\lambda'_{w'})\pi_Y(v)$, and so 
	$\pi_{Y'}(w')=\sum _{v\in W'} (-\lambda'_v/\lambda'_{w'})\pi_{Y'}(v)$. The uniqueness of $(\lambda_v)_{v\in W'}$ ensures that $-\lambda'_v/\lambda'_{w'}=\lambda_v$ for each 
	$v\in W'$. But $w'_\alpha\ne \sum_{v\in W'} \lambda_v v_\alpha=\sum _{v\in W'} (-\lambda'_v/\lambda'_{w'})v_\alpha$, a contradiction. 
	
	Thus the required claim holds for $q$.
        \qed
\end{proof} 

We recall the following facts from the OEIS \cite{IntegerSequenceEncyclopedia}.
\begin{sloppypar}
For each natural $n$, let $a(n)$ be the maximal determinant of a matrix of order $n$ whose entries are $0$ or $1$. The values of $a(n)$ are known up to { $n=21$: $1, 1, 2, 3, 5, 9, 32, 56, 144, 320, 1458, 3645, 9477, 25515, 131072, 327680, 1114112, 3411968, 19531250, 56640625, \linebreak 195312500,\dots$ }. Hadamard proved that $a(n) \le 2^{-n}(n+1)^{\frac{n+1}{2}}$, with equality iff a Hadamard matrix of order $n+1$ exists. It is believed that the latter holds iff $n+1 = 1, 2$ or a multiple of $4$.
\cite{clements1965sequence} provide a lower bound
$a(n) > 2^{-n}(\frac{3}{4}(n+1))^{\frac{n+1}{2}}$.

Here are two examples of $6\times 6$ matrices that attain the upper bound $a(6)=9$:%
\footnote{
The first is from OEIS;
the second is from Dietrich Burde in
\url{https://math.stackexchange.com/a/4965827}.
}
\begingroup
\setlength{\arraycolsep}{5pt}
\begin{align}
\label{eq:mat9}
\begin{bmatrix} 
  1 & 0 & 0 & 1 & 1 & 0
  \\
  0 & 0 & 1 & 1 & 1 & 1
  \\
  1 & 1 & 1 & 0 & 0 & 1
  \\
  0 & 1 & 0 & 1 & 0 & 1
  \\
  0 & 1 & 0 & 0 & 1 & 1
  \\
  0 & 1 & 1 & 1 & 1 & 0
\end{bmatrix}   
&&
\begin{bmatrix} 
  1 & 0 & 1 & 0 & 0 & 0
  \\
  1 & 1 & 0 & 1 & 0 & 0
  \\
  0 & 1 & 1 & 0 & 1 & 0
  \\
  0 & 0 & 1 & 1 & 0 & 1
  \\
  1 & 0 & 0 & 1 & 1 & 0
  \\
  1 & 1 & 0 & 0 & 1 & 1
\end{bmatrix}   
\end{align}
\endgroup

\end{sloppypar}

\begin{lemma}
	\label{integer coefficeints lemma}
	Let $Z \subset \{0,1\}^X$ be a nonempty finite linearly dependent set of nonzero vectors and
	 $p = |Z| - 1 $. Then there exist integers $(\Delta_w)_{w \in Z}$ which are not all zeros such that $|\Delta_w| \leq a(p)$ for each $w \in Z$ and $\sum_{w \in Z} \Delta_w \cdot w = 0$.
	That is if \emph{some} nontrivial linear combination of $Z$ equals $0$, then there exists such a linear combination in which the coefficients are all integers, and are all bounded by $a(p)$.
\end{lemma}
\begin{proof}
Let $W$ be a maximal linearly independent subset of $Z$ (note that $|W|\leq |Z|-1=p$). By \Cref{lem:Y},  there exists a subset $Y$ of $X$ with $|Y|= |W|$ such that the set $\pi_Y(W) \subset \mathbb{R}^Y$ is linearly independent.
	
Pick any vector $z \in Z \setminus W$. By the maximality of $W$, the set $W\cup\{z\}$ is linearly dependent, and so the set $\pi_Y(W)\cup\{\pi_Y(z)\}$ is linearly dependent too.
Therefore, the system $\sum_{v\ \in W} x_v \cdot \pi_Y(v)=\pi_Y(z)$ has a solution $\mathbf{x}$. Note that this is a square system, with $|Y|$ equations in $|W|=|Y|$ unknowns.

The solution $\mathbf{x}$ is unique, because $\pi_Y(W)$ is linearly independent.
Therefore, $\mathbf{x}$ can be computed by Cramer's rule.
It gives $x_v= \Delta_v/\Delta_z$ for each $ v \in W$, where $\Delta_z$ is the determinant of the matrix with columns $\pi_Y(v)$, and $\Delta_v$ 
is the determinant of the same matrix where column $v$ is replaced with $\pi_Y(z)$. 
All these matrices are square matrices with binary entries and size at most $|Y|\leq p$, so their determinants are integers with absolute value at most $a(p)$.

Since the set $W\cup \{z\}$ is linearly dependent and the system $ \sum_{v \in W} x_v \cdot \pi_Y(v)=\pi_Y(z)$ has a unique solution, the system $ \sum_{v \in W} x_v\cdot v=z$ has (the same) unique solution. 
Multiplying by $\Delta_z$ gives
$\sum_{v \in W} \Delta_v\cdot v - \Delta_z\cdot z = 0$
Putting $ \Delta_v=0$ for each $v\in Z \setminus (W\cup\{z\})$ yields the desired linear combination.
\qed
\end{proof}

Given an input set of items $D$, let $OPT(D_k)$ denote the optimal number of bins in $k$-times bin-packing of the items in $D$.
\begin{theorem} 
\label{theorem: finite-k-exists}
For any electricity division instance with 
demand-vector $D$  and supply $S$, 
there exists an integer $k \leq a(n)$ 
such that $\frac{k}{OPT(D_k)}$ is the egalitarian connection-time per agent.
This time can be attained by solving \kbp{} on $D$ and allocating a fraction $\frac{1}{OPT(D_k)}$ of the time to each bin in the optimal solution.
\end{theorem} \label{existence of finite k}
\begin{proof}
	The set $W = \{ (w_1, \ldots, w_n) \in \{  0,1 \}^n : \sum_{i=1}^{n} d_i w_i \leq S \}$ naturally represents all admissible ways to pack the agents into bins (all feasible ``configurations''). The convex hull $\mathrm{CH}(W)$ of $W$ naturally represents the set of all possible schedules, where a schedule is represented by the fraction of time allocated to each configuration.
	
	Denote $e=(1, \ldots, 1)\in\mathbb{R}^n$. Let $r_{\max}$ be the egalitarian connection time, defined by $r_{\max} = \mathrm{sup} \{ r \in [0,1] : r\cdot e \in \mathrm{CH}(W) \}$. Since the set $\mathrm{CH}(W)$ is compact, $r_{\max}$ is attained, that is $r_{\max}\cdot e \in \mathrm{CH}(W)$. The electricity division problem has a solution, so $r_{\max} > 0$. 
	
	Since $r_{\max}\cdot e \in \mathrm{CH}(W)$, by Caratheodory's theorem \cite{enwiki:1216835967_CaratheodoryThm}, there exists a simplex $\mathrm{CH}(W')$ with the set $W'\subset W$ of vertices such that $r_{\max}\cdot e \in \mathrm{CH}(W')$.	
	
	Consider the following illustrating example. Let the supply $S$ is $25$ and there are three agents $x$, $y$, and $z$ with the demands $d_x = 11$, $d_y = 12$, and $d_z = 13$, respectively. It is possible to connect each agent per $r_{\max} = 2/3$ of the time by connecting each of $\{x,y\}, \{y,z\}, \{x,z\}$ per $1/3$ of the time. This can be visually represented as shown at figure \ref{fig:two-third-example}.
	\begin{figure}[ht]
		\centering
		\includegraphics[height=0.32\textheight, width=0.32\linewidth, keepaspectratio]{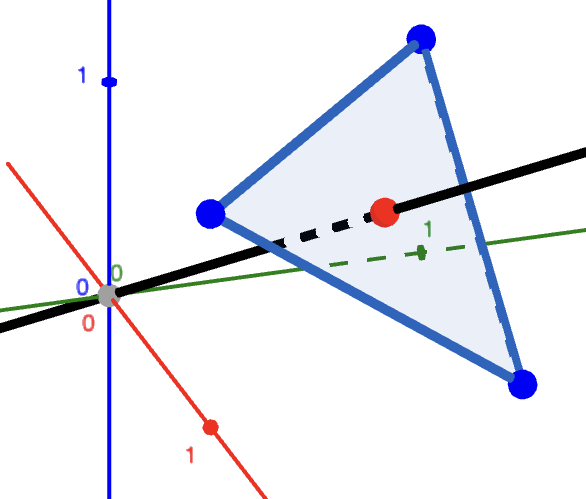}
            \caption{The three blue points are (1,1,0), (1,0,1) and (0,1,1). The black line that originates from the origin intersects with the triangle (convex hull of the three blue points) at red point (2/3,2/3,2/3).}
		\label{fig:two-third-example}
	\end{figure}
	The three blue points $(1,1,0)$, $(1,0,1)$, and $(0,1,1)$ represent elements of $W'$, and the white triangle with these vertices represents $\mathrm{CH}(W')$. The black ray originating from the origin $(0,0,0)$ is $e\cdot \{ r: r\geq 0 \}$. The red point $(2/3, 2/3, 2/3)$   represents the intersection of the ray with $\mathrm{CH}(W')$ and corresponds to $r_{\max} = 2/3$.

	Let $W''$ be the set of the vertices of the face of $\mathrm{CH}(W')$ of minimal dimension containing $r_{\max}\cdot e$. Clearly, $r_{\max} \cdot e$ is a boundary point of $\mathrm{CH}(W')$,  so $|W''| \leq n$. In other words, there exists an optimal electricity division schedule with at most $n$ different configurations. In the above example $W'' = W'$ and  $|W''| = 3 = n$.
	
	Let $(\lambda_w)_{w \in W''}$ be the barycentric coordinates of $r_{\max}\cdot e$, such that $\lambda_w > 0$ for each $w \in W''$ and $\sum_{w \in W''}\lambda_w = 1$, and
	\begin{align}
	\label{eq:rmaxe}
	r_{\max}\cdot e = \sum_{w \in W''} \lambda_w w.    
	\end{align}
	In the above example $\lambda_w = 1/3$ for each $w \in W''$.
	
	We claim that the set $W''$ is linearly independent. Indeed, suppose for a contradiction that there exist disjoint nonempty subsets $W''_1$ and $W''_2$ of $W''$ and positive numbers $(\mu_w)_{w \in W''_1 \cup W''_2}$ such that $ \sum_{w \in W''_1} \mu_w w = \sum_{w \in W''_2} \mu_w w $.
	Assume w.l.o.g. that $ \sum_{w \in W''_1} \mu_w \leq \sum_{w \in W''_2} \mu_w $. Let $ \nu = \min_{w \in W''_2} {\lambda_w / \mu_w}
	$. Then,
	\[
	r_{\max}\cdot e = \sum_{w \in W''} {\lambda_w w} = 
	\sum_{w \in W''} {\lambda_w w}  + \nu \left(\sum_{w \in W''_1} {\mu_w w} - \sum_{w \in W''_2} {\mu_w w}\right) = \sum_{w \in W''} {\lambda'_w w},
	\]
	where
	\begin{equation}
		\lambda'_w  = 
		\begin{cases}
			\lambda_w + \nu \mu_w	& \text{ if } w \in W''_1 \\	
			\lambda_w - \nu \mu_w	& \text{ if } w \in W''_2 \\
			\lambda_w 	& \text{ if } w \in W'' \setminus (W''_1 \cup W''_2)
		\end{cases} 
	\end{equation}
	Then, $ s := \sum_{w \in W''} {\lambda'_w} = \sum_{w \in W''} {\lambda_w}  + \nu (\sum_{w \in W''_1} {\mu_w} - \sum_{w \in W''_2} {\mu_w}) \leq 1, \lambda'_w \geq 0 $ for each $ w \in W'' $, and $ \lambda'_w = 0$ for some $ w \in W''_2$. Then $ s^{-1} \cdot r_{\max} \cdot e \in \mathrm{CH} (W'' \setminus \{w\})$, that contradicts the minimality of $W''$.

Let $Z:= W'' \cup \{e\}$.
As $W''$ contains a multiple of $e$, the set $Z$ is linearly dependent. 
\Cref{integer coefficeints lemma} can be applied to it with $p = |W''|\leq n$.
The lemma implies that there exist integers $(\Delta_v)_{v \in Z}$ which are not all zeroes such that $ |\Delta_v| \leq a(p)\leq a(n)$
 for each $v \in Z$ and $ \sum_{v \in Z} \Delta_v \cdot v = 0$. The set $W''$ is linearly independent, so $ \Delta_e \neq 0$ and 
\begin{align}
\label{eq:esum}
e = \sum_{w \in W''} {(-\Delta_w / \Delta_e)w}.     
\end{align}
But by \eqref{eq:rmaxe},
$e = \sum_{w \in W''} {(\lambda_w/r_{\max})}\cdot w $. As $W''$ is linearly independent, the coefficients in the expression for $e$ are unique, so we must have $-\Delta_w / \Delta_e = \lambda_w / r_{\max} > 0 $ for each $w \in W''$. 
Therefore, $-\Delta_w / \Delta_e = |\Delta_w| / |\Delta_e|$.

Construct a packing in which each configuration $w\in W''$ appears $|\Delta_w|$ times.
Then the vector 
$\sum_{w \in W''} {|\Delta_w|\cdot w}$ represents the number of times each item appears in the packing; but by \eqref{eq:esum}, this sum equals $|\Delta_e|\cdot e$. Therefore, it is a $k$-times bin-packing with $k := |\Delta_e|\leq a(n)$.

The total number of bins in the packing is $\sum_w |\Delta_w| = 
\sum_w |\Delta_e|\cdot \lambda_w/r_{\max}
=k/r_{\max}$.
Therefore, connecting each bin for an equal amount of time yields an allocation in which each agent is connected for a fraction $k/(k/r_{\max}) = r_{\max}$ of the time, as required.
\qed
\end{proof}

\subsection{Lower bound}
For any electricity division instance $D$,
let $K(D)$ denote the smallest $k$ such that 
the egalitarian connection time equals $\frac{k}{OPT(D_k)}$.
For any $n$, let $K(n)$ denote the maximum $K(D)$ over all instances with $n$ households.
\Cref{theorem: finite-k-exists} implies $K(n)\leq a(n)$, which provides an exponential upper bound on $K(n)$.
This raises the question of whether there is a matching lower bound.
Currently, we only have two very loose lower bounds.

\begin{proposition}
\label{prop:lower:n-1}
For any $n\geq 2$, there is a lower bound $K(n)\geq n-1$
\end{proposition}
\begin{proof}
Consider an instance with $n$ items of size $1$, and let $S:=n-1$. 
Then we can construct $n$ bins, each of which contains a different subset of $n-1$ items. 
This packing is optimal, as all bins are full; it is an $(n-1)$-times bin-packing, as each item appears in exactly $n-1$ different bins.
Therefore, $r_{\max} = (n-1)/n$.
We show that any $k<n-1$ does not yield an optimal packing. For any $k$, the sum of all item sizes is $k n$, so any $k$-times bin-packing must contain at least $k n / S = k\frac{n}{n-1}$ bins. But $n$ and $n-1$ are coprime, so for any $k<n-1$, the expression $k\frac{n}{n-1}$ is not an integer, which means that the packing must have non-full bins. \qed
\end{proof}

\begin{proposition}
\label{prop:lower:9}
$K(6) = 9$.
\end{proposition}
\begin{proof}
As $K(6)\leq a(6)=9$, it is sufficient to prove the lower bound $K(6)\geq 9$.

Let $A$ be the rightmost $6\times 6$ matrix in \eqref{eq:mat9}. Let $S := \det(A) =9$.

Let $B := \det(A)\cdot A^{-1}$, and note that $B$ is an integer matrix.

Let the vector of demands be $D := B\cdot e = [4,2,5,3,2,1]$.
By construction, Each row in $A$ corresponds to a configuration with sum exactly $9$. These configurations are: $4+5, 4+2+3, 2+5+2', 5+3+1, 4+3+2', 4+2+2'+1$ (where $2'$ denotes the second household with demand $2$).

Construct a packing by $\mathbf{x} := B^T \cdot e = [1,2,3,5,2,4]$.
Every element $x_i$ represents the number of times that configuration $i$ appears in the packing. In particular, there are ---
\begin{itemize}
\item $1$ bin with $4+5$;
\item $2$ bins with $4+2+3$;
\item $3$ bins with $2+5+2'$;
\item $5$ bins with $5+3+1$;
\item $2$ bins with $4+3+2'$;
\item $4$ bins with $4+2+2'+1$.
\end{itemize}
By construction, each item appears exactly $9$ times, so it is a valid $9$-times bin-packing. It has $17$ bins, and it is optimal since all bins are full. Therefore, $r_{\max} = 9/17$.

We now show that any $k<9$ does not yield an optimal packing. For any $k$, the sum of all bins in $k$BP is $17 k$, so any $k$BP must contain at least $17 k / 9$ bins. 
But $17$ and $9$ are coprime, Therefore, for any $k<9$, the expression $17 k / 9$ is not an integer, which means that the packing must have non-full bins. \qed
\end{proof}
It is open whether \Cref{prop:lower:9} can be generalized.
The following conjecture aims to generalize \Cref{prop:lower:9}.

\begin{conjecture}
\label{prop:coprime}
Let $A$ be an $n\times n$ binary matrix, 
let $B := \det(A)\cdot A^{-1}$,
and let $g$ be the sum of all elements in $B$ (the ``grand sum'' of $B$).
If $g$ and $\det(A)$ are coprime, then $K(n)\geq \det(A)$.
\end{conjecture}
Note that \Cref{prop:lower:9} is a special case with $\det(A)=9$ and $g=17$.
\section{Fast Approximation Algorithms} \label{section:approx_algo}
The results of the previous section are not immediately applicable to fair electricity division, as $k$BP is known to be an NP-hard problem. However, they do hint that good approximation algorithms for $k$BP can provide good approximation for electricity division.
Therefore, in this section, we study several fast approximation algorithms for $k$BP.

\subsection{\texorpdfstring{$FFk$}{FFk} --- First-Fit for \texorpdfstring{$k$}{k}BP} \label{section:approx-algo:subsection:approx_algo:ffk}
The $k$-times version of the First-Fit bin-packing algorithm packs each item of $D_k$ in order into the first bin where it fits and does not violate the constraint that each item should appear in a bin at most once. If the item to pack does not fit into any currently open bin, $FFk$ opens a new bin and packs the item into it. For example: consider $D = \{10,20,11\}, k=2, S=31$. $FFk$ will result the bin-packing $\{10,20\}, \{11,10\}, \{20,11\}$. 
It is known that the asymptotic approximation ratio of $FF$ is $1.7$ \cite{Dsa2013}.
Below, we prove that, for any fixed $k>1$, the asymptotic approximation ratio of $FFk$ for large instances (when $OPT(D_k)\to \infty$) is better, and it improves when $k$ increases.

\begin{theorem} \label{ffk:theorem3/2}
	For every input $D$ and $k\geq 1$, $FFk(D_k) \leq \left(1.5+\frac{1}{5k}\right) \cdot OPT(D_k) + 3\cdot k$.
\end{theorem}

\begin{proof}
	At a very high level, the proof works as follows:
	\begin{itemize}
		\item We define a \emph{weight} for each item, which depends on the item size, and may also depend on the instance to which the item belongs.
		\item We prove that the average weight of each bin in an optimal packing is at most some real number $Z$, so the total weight of all items is at most $Z\cdot OPT(D_k)$.
		\item We prove that the average weight of a bin in the $FFk$ packing (except some $3 k$ bins that we will exclude from the analysis) is at least some real number $Y$, so the total weight of all items is at least $Y\cdot (FFk(D_k) - 3 \cdot k)$.
		\item Since the total weight of all items is fixed, we get $FFk(D_k) - 3\cdot k \leq (Z/Y)\cdot OPT(D_k)$, which gives an asymptotic approximation ratio of $Z/Y$.
	\end{itemize}
	
	\paragraph{\textbf{Basic weighting scheme}}
	
	The basic weighting scheme we use follows \cite{Dsa2013}. The weight of any item of size $v$ is defined as
	\begin{align*}
		w(v) := v/S + r(v),
	\end{align*}
	where $r$ is a \emph{reward function}, computed as follows:
	\[
	r(v) := 
	\begin{cases}
		0 \quad & \text{if } v/S \leq \frac{1}{6}, \\
		\frac{1}{2}(v/S - \frac{1}{6}) \quad & \text{if } v/S \in (\frac{1}{6}, \frac{1}{3}), \\
		1/12 \quad & \text{if } v/S \in [\frac{1}{3}, \frac{1}{2}], \\
		4/12 \quad & \text{if } v/S > \frac{1}{2}.
	\end{cases}
	\]

	For $k>1$, we use a modified weighting scheme, which  gives different weights to items that belong to different instances. This allows us to get a better asymptotic ratio. We describe the modified weighting scheme below.
	
	\paragraph{\textbf{Associating bins with instances}}
	
	Recall that $FFk$ processes one instance of $D$ completely, and then starts to process the next instance. 
	We associate each bin in the $FFk$ packing with the instance in which it was opened. So for every $j \in\{1,\ldots,k\}$, the \emph{bins of instance $j$} are all the bins, whose first allocated item comes from the $j$-th instance of $D$.
	Note that bins of instance $j$ do not contain items of  instances $1,\ldots,j-1$, but may contain items of any instance $j,\ldots,k$. 
	
	For the analysis, we also need to associate some bins in the optimal packing with specific instances.  We consider some fixed optimal packing.
	For each bin $B$ in that packing:
	\begin{itemize}
		\item If $B$ contains an item $x$ with $V(x)>S/2$, and $x$ belongs to instance $j$, 
		then we associate bin $B$ with instance $j$. Clearly, there can be at most one item of size larger than $S/2$ in any feasible bin, so there is no ambiguity.
		Moreover, we ensure that all other items in $B$ belong to instance $j$ too: if some other item $x'\in B\setminus \{x\}$ belongs to a different instance $j'$, then we replace $x'$ with its copy from instance $j$. Note that the other copy of $x'$ must be located in a bin different than $B$, due to the restrictions of $k$BP. Since both copies of $x'$ have the same size, the size of all bins remains the same.
		\item If $B$ contains no item of size larger than $S/2$, then we do not associate $B$ with any instance. 
	\end{itemize}
	
	\paragraph{\textbf{Partitioning $FFk$ bins into groups}}
	For the analysis, we partition the bins of each instance $j\in[k]$ in the $FFk$ packing into five groups.
	
	\textbf{Group 1}.
	Bins with a single item of instance $j$, whose size is at most $S/2$. 
	There is at most one such bin.
	This is because, if there is one such bin $B$ of instance $j$,
	it means that all later items of instance $j$ do not fit into $B$, so their size must be greater than $S/2$. Therefore, any later bin $B'$ of instance $j$ must have size larger than $S/2$.
	
	\textbf{Group 2}.
	Bins with two or more items of instance $j$, whose total size is at most $2S/3$.
	There is at most one such bin.
	This is because, if there is one such bin $B$ of instance $j$,
	it means that all later items of instance $j$ do not fit into $B$, so their size must be greater than $S/3$. Therefore, in any later bin $B'$ with two or more items of instance $j$, their total size is larger than $2S/3$.
	
	\textbf{Group 3}.
	Bins with a single item of instance $j$, whose size is larger than $S/2$. 
	
	\textbf{Group 4}.
	Bins with two or more items of instance $j$, whose total size is at least $10 S/12$.
	
	\textbf{Group 5}.
	Bins with two or more items of instance $j$, whose total size is in $(2 S/3, 10 S/12)$.
	
	In the upcoming analysis, we will exclude from each instance, the at most one bin of group 1, at most one bin of group 2, and at most one bin of group 5. All in all, we will exclude at most $3k$ bins of the $FFk$ packing from the analysis.
	
	\paragraph{\textbf{Analysis using the basic weighting scheme}}
	First we prove that, with the basic weighting scheme, the total weight of each optimal bin $B$ is at most $17/12$. Since the total size of $B$ is at most $S$, we have $w(B)\leq 1+r(B)$, so it is sufficient to prove that the reward $r(B)\leq 5/12$. Indeed:
	\begin{itemize}
		\item If $B$ contains an item larger than $S/2$, then 
		this item gives $B$ a reward of $4/12$;
		the remaining room in $B$ is smaller than $S/2$. 
  This can accommodate either a single item of size at least $S/3$ and some items smaller than $S/6$,
  or two items of size between $S/6$ and $S/3$; in both cases, the total reward is at most $1/12$. 
            \item If $B$ does not contain an item larger than $S/2$, then there are at most $5$ items larger than $S/6$, so at most $5$ items with a positive reward. The reward of each item of size at most $S/2$ is at most $1/12$.
	\end{itemize}
 In both cases, the total reward is at most $5/12$.	
 
	We now prove that, with the basic weighting scheme, the average weight of each $FFk$ bin is at least $10/12$. 
	In fact, we prove a stronger claim: we prove that, for each instance $j$, the average weight of the items \emph{of instance $j$ only} in bins of instance $j$ (ignoring items of instances $j+1,\ldots k$ if any) is at least $10/12$.
	
	We use the above partition of the bins into 5 groups, excluding at most one bin of group 1 and at most one bin of group 2.
	
	The bins of group 3 (bins with a single item of instance $j$, which is larger than $S/2$) have a reward of $4/12$, so their weight is at least $1/2+4/12 = 10/12$.
	
	The bins of group 4 (bins with two or more items of instance $j$, which have a total size larger than $10 S/12$) already have weight at least $10/12$, regardless of their reward.
	
	We now consider the bins of group 5 (bins with two or more items of instance $j$, which have a total size in $(2S/3, 10S/12)$).
	Denote the bins in group 5 of instance $j$, in the order they are opened, by $B_1,\ldots,B_Q$.
	Consider a pair of consecutive bins, for instance, $B_t, B_{t+1}$.
	Let $x := $ the free space in bin $B_t$ during instance $j$.
	Note that $x\in (S/6, S/3)$. Let $c_1,c_2$ be some two items of instance $j$ which are packed into $B_{t+1}$.
	
	For each $c_i \in \{c_1,c_2\}$, as $FFk$ did not pack $c_i$ into $B_t$ during the processing of instance $j$, there are two options:
	either $c_i$ is too large (its size is larger than $x$), or $B_t$ already contains other copies of $c_i$ from other instances. But the second option cannot happen, because $B_t$ belongs to instance $j$, so it contains no items of instances $1,\ldots,j-1$; and while $c_i$ of instance $j$ was processed by $FFk$, items of instance $j+1,\ldots,k$ were not processed yet. Therefore, necessarily $V(c_i)>x$.
	
	Since $S/6<x<S/3$, the reward of each of $c_1,c_2$ is at least $(x/S-1/6)/2$, and the reward of both of them is at least $x/S-1/6$. Therefore, during instance $j$,
	\begin{align*}
		V(B_t)/S + r(B_{t+1}) \geq (S-x)/S + (x/S-1/6)
		= 1-1/6 = 10/12.
	\end{align*}
	In the sequence $B_1,\ldots,B_Q$, there are $Q-1$ consecutive pairs. Using the above inequality, we get that the total weight of these bins is at least $10/12\cdot (Q-1)$, which is equivalent to excluding one bin (in addition to the two bins excluded in groups 1 and 2).
	
	Summing up the weight of all non-excluded bins gives at least $(10/12)\cdot (FFk(D_k)-3k)$. Meanwhile, the total weight of optimal bins is at most $(17/12)\cdot OPT(D_k)$. This yields $FFk(D_k)\leq (17/10)\cdot OPT(D_k) + 3 k$, which corresponds to the known asymptotic approximation ratio of $1.7$ for $k=1$.
	
	We now present an improved approximation ratio for $k>1$. 
	To do this, we give different weights to items of different instances. 
	We do it in a way that the average weight of the $FFk$ bins (except the $3k$ excluded bins) will remain at least $10/12$. 
	Meanwhile, the average weight of the optimal bins in some $k-1$ instances (as well as the unassociated bins) will decrease to $15/12$, whereas the average weight of the optimal bins in a single instance will remain $17/12$.
	This would lead to an asymptotic ratio of at most $1.5+\frac{1}{5k}$.

	\paragraph{\textbf{Warm-up: $k=2$}}
	As a warm-up, we analyze the case $k=2$, and prove an asymptotic approximation ratio of $1.5+\frac{1}{2\cdot 5} = 1.6$.
	We give a detailed proof of the general case of any $k \geq 1$ in appendix \ref{appendix: ffk-algorithm: ffk-general-case-proof}. 
	
	We call a bin in the $FFk$ packing \emph{underfull} if it has only one item, and this item is larger than $S/2$ and smaller than $2S/3$. Note that all underfull bins belong to group 3. 
	We consider two cases.
	
	\paragraph{Case 1} No $FFk$ bin of instance $1$ is underfull (in instance $2$ there may or may not be  underfull bins). 
	In this case, we modify the weight of items in instance $1$ only: we reduce the reward of each item of size $v>S/2$ in instance $1$ from $4/12$ to $2/12$. 
	\begin{itemize}
		\item For any bin $B$ in the optimal packing, if $B$
		does not contain an item larger than $S/2$, then its maximum possible reward is still $3/12$ as with the basic weights. 
		If $B$ contains an item larger than $S/2$, then this item gives $B$ a reward of $2/12$ if it belongs to instance $1$, or $4/12$ if it belongs to instance $2$. The remaining room in $B$ is smaller than $S/2$, and the maximum total reward of items that can fit into this space is $1/12$ 
(as the remaining space can accommodate either one item of size at least $S/3$ and some items of size smaller than $S/6$, or two items of sizes in $(S/6,S/3)$). 
Overall, the reward of $B$ is at most $5/12$ if it contains an item larger than $S/2$ from instance $2$, and at most $3/12$ otherwise. As at most half the bins in the optimal packing contain an item larger than $S/2$ from instance $2$, at most half these bins have reward $5/12$, while the remaining bins have reward at most $3/12$, so the average reward is at most $4/12$.
		Adding the size of at most $1$ per bin leads to an average weight of at most $16/12$.
		\item In the $FFk$ packing, we note that the change in the reward affects only the bins of group $3$ (bins with a single item and $V(B)>S/2$). These bins now have a reward of at least $2/12$. Since none of them is underfull by assumption, their weight is at least $2/3+2/12 = 10/12$.
		The bins of groups $4$ and $5$ are not affected by the reduced reward: the same analysis as above can be used to deduce that their average weight (except one bin per instance excluded in group $5$) is at least $10/12$.
	\end{itemize}
	
	\paragraph{Case 2} At least one $FFk$ bin of instance $1$ is underfull.
	In this case, we modify the weight of items in instance $2$ only, as follows  (note that we reduce the \emph{weights} of these items, and not only their rewards):
	\[
	w(v) = 
	\begin{cases}
		0 \quad & \text{if } v/S < \frac{1}{3}, \\
		5/12 \quad & \text{if } v/S \in [\frac{1}{3}, \frac{1}{2}], \\
		10/12 \quad & \text{if } v/S > \frac{1}{2}.
	\end{cases}
	\]
	Note that all these weights are not higher than the basic weights of the same items. 
	\begin{itemize}
		\item In the optimal packing, every bin that contains an item larger than $S/2$ from instance $2$ is associated with instance $2$, and therefore contains only items of instance $2$ by construction. Therefore, the weight of any such bin is at most $15/12$ (the remaining space in such bin can pack only one item from $[S/3,S/2]$ of positive weight).
		The weight of every bin that contains no item larger than $S/2$ is still at most $15/12$, since its reward is at most $3/12$ as with the basic weights.
		The only bins that may have a larger weight (up to $17/12$) are those that contain an item larger than $S/2$ from instance $1$; at most half the bins in the optimal packing contain such an item. Therefore, the average  weight of a bin in the optimal packing is at most $16/12$.
		\item In the $FFk$ packing, 
		some bin $B^1$ in instance $1$ is underfull, that is, $B^1$ contains a single item of size $v\in(S/2, 2S/3)$, and the free space in $B^1$ is $S-v\in (S/3, S/2)$. 
		This means that, in the bins of instance $2$, there is no item of size at most $S/3$, since every such item would fit into $B^1$ (it is smaller than the available space in $B^1$, and it is not a copy of the single item in $B^1$).
		So for every bin $B^2$ in instance $2$ (except the excluded ones), there are only two options:
		\begin{itemize}
			\item $B^2$ contains one item larger than $S/2$ --- so its weight is $10/12$;
			\item $B^2$ contains two items, each of which is larger than $S/3$ --- so its weight is at least $2\cdot 5/12 = 10/12$.
		\end{itemize}
		The weight of items in instance $1$ does not change. So the average weight of bins in the $FFk$ packing (except the $3k$ excluded bins) remains at least $10/12$.
	\end{itemize}
	
	In all cases, for $k=2$, we get that the average weight of bins in the optimal packing is at most $16/12$, and the average weight of non-excluded bins in the $FFk$ packing is at least $10/12$. Therefore,
	$FFk -3k \leq (16/10)\cdot OPT = 1.6\cdot OPT$.
\qed	
\end{proof}

\paragraph{Approximation ratio lower bound}
For $k=1$, the approximation ratio of $1.7$ is tight; an example is given in \cite{Dósa_Sgall_2014} and on page 306 in \cite{Johnson1974}.
In appendix \ref{appendix: ffk-worst-case-example-johnsonpaper}, we analyze these examples for $k=2$, and show that $FFk$ achieves a ratio of $1.35$. For the example given in \cite{Johnson1974} we observed that as $k$ increases, the approximation ratio continues to decrease. This raises the question of whether the bound of Theorem \ref{ffk:theorem3/2} is tight.

Although we do not have a complete answer for this question, we show below an example in which $FFk$ attains an approximation ratio of $1.375$ for $k=2$, which is higher than the example of \cite{Johnson1974}.

Consider the example $D = \{371, 659, 113, 47, 485, 3, 228, 419, 468, 581, 626 \}$ and bin capacity ${S} = 1000$. Then, $FFk$ for $k=2$ will result in 11 bins $[371, 113, 47, 3, 228],$ $[659, 113, 47, 3],$ $[485, 419],$ $[468, 371],$ $[581, 228],$ $[626], [659],$ $[485, 419],$ $[468], [581], [626]$. The optimal packing for the items in $D_2$ is $[626, 371, 3],$ $[626, 371, 3]$, $[659, 228, 113],$ $[659, 228, 113],$ $[419, 581],$ $[419, 581],$ $[468, 485, 47],$ $[468, 485, 47]$. Clearly, $\frac{FFk(D_2)}{OPT(D_2)} = \frac{11}{8} = 1.375$. 
We observed that as $k$ increases, the approximation ratio continues to decrease.
Experimental results in support of this ratio are given in appendix \ref{section:results}.
Based on the example for which we have achieved a ratio of $1.375$ and the simulation of the algorithm on different datasets, we conjecture that for $k>1$ the \emph{absolute} approximation ratio for $FFk$ is $1.375$. 

\subsection{\texorpdfstring{$FFDk$}{FFDk}} \label{section:approx-algo:subsection:approx_algo:ffdk}
The $k$-time version of the First-Fit Decreasing bin-packing algorithm first sorts $D$ in non-increasing order. Then it  constructs $D_k$ using $k$ consecutive copies of the sorted $D$, and then implements $FFk$ on  $D_k$.
In contrast to $FFk$, we could not prove an upper bound for $FFDk$ that is better than the upper bound for $FFD$; we only have a lower bound.

\begin{lemma} \label{ffdk:lemma1}
	$FFDk(D_k) \geq \frac{7}{6} \cdot OPT(D_k) + 1$.
\end{lemma}


\begin{proof}
	We use the following example from page $2$ of \cite{dosa_tight_2007}. Let $\delta$ be a sufficiently small positive number, and let $S=1$. Let $D = \{\frac{1}{2}+ \delta, \frac{1}{2}+ \delta, \frac{1}{2}+ \delta, \frac{1}{2}+ \delta, \frac{1}{4}+ 2\delta, \frac{1}{4}+ 2\delta, \frac{1}{4}+ 2\delta, \frac{1}{4}+ 2\delta, \frac{1}{4}+ \delta, \frac{1}{4}+ \delta, \frac{1}{4}+ \delta, \frac{1}{4}+ \delta, \frac{1}{4} - 2\delta, \frac{1}{4} - 2\delta, \frac{1}{4} - 2\delta, \frac{1}{4} - 2\delta, \frac{1}{4} - 2\delta, \frac{1}{4} - 2\delta, \frac{1}{4} - 2\delta, \frac{1}{4} - 2\delta\}$. 
 
 An optimal packing for $D$ contains 4 bins of type $\{ \frac{1}{2}+\delta, \frac{1}{4}+\delta, \frac{1}{4} - 2\delta \}$ and 2 bins of type $\{ \frac{1}{4}+ 2\delta, \frac{1}{4}+ 2\delta, \frac{1}{4} - 2\delta, \frac{1}{4} - 2\delta \}$. As all bin sizes are exactly $1$, this pattern is clearly optimal for any $k$. Therefore, for all $k\geq 1$, $OPT(D_k) = 6k$. 
	
	On applying $FFDk$ on $D_k$ the resulting number of bins are $8 + 7(k-1)$. This gives us a lower bound of $\frac{7}{6} \cdot OPT(D_k) + 1$.
 \qed
\end{proof}

We discuss the challenges in extending the existing proof for the $FFD$ in appendix \ref{appendix:FFDk-algorithm}.
Based on the simulation of $FFDk$ on different datasets, we conjecture that the upper bound for $FFDk$ is $\frac{11}{9} \cdot OPT(D_k) + \frac{6}{9}$. Experimental results supporting this conjecture are provided in  \Cref{section:results}.

\subsection{\texorpdfstring{$NFk$}{NFk}}
\label{section:approx-algo:subsection:NFk}
Given the input $D_k$, the algorithm $NFk$ works as follows: like ${NF}$, ${NFk}$ always keeps a single bin open to pack items. If the current item does not pack into the currently open bin then ${NFk}$ closes the current bin and opens a new bin to pack the item. 

We can assume that $V(D) > S$, otherwise there is a trivial solution with $k$ bins. 
While processing input $D_k$, $NFk$ holds only one open bin, and it cannot contain a copy of each item of $D$. In fact, the open bin always contains a part of some instance of $D$, and possibly a part of the next instance of $D$, with no overlap. Therefore, 
if the current item $x$ is not packed into the current open bin, the only reason is that $x$ does not fit, as there is no  previous copy of $x$ in the current bin (all previous copies, if any, are in already-closed bins).
\begin{theorem} \label{NFk:theorem:asymptotic-ratio-2}
	For every input $D_k$ and $k \geq 1$,  the asymptotic ratio of ${NFk}(D_k)$ is 2.
\end{theorem}
We give a detailed proof of the above Theorem in appendix \ref{appendix:NFk-algorithm}.

\section{Polynomial-time Approximation Schemes} \label{section:ptas}
\subsection{Some general concepts and techniques} \label{section:EAA:subsection:general}
The basic idea behind generalizing Fernandez de la Vega-Lueker and all the Karmarkar-Karp algorithms to solve $k$BP is similar. It consists of three steps: 
1. Keeping aside the small items, 2. Packing the remaining large items, and 3. Packing the small items in the bins that we get from step 2 (opening new bins if necessary) to get a solution to the original problem.

In step 3, the main difference from previous work is that, in $k$BP, we cannot pack two copies of the same small item into the same bin, so we may have to open a new bin even though there is still remaining room in some bins. The following lemma analyzes the approximation ratio of this step.
\begin{lemma} \label{general:lemma1}
	Let $D_k$ be an instance of the $k$BP problem, and $0 < \epsilon \leq 1/2$.
We say that the item is \emph{large}, if its size is bigger than $\epsilon \cdot S$ and \emph{small} otherwise. 
Assume that the large items are
packed into $L$ bins. Consider an algorithm which starts adding the small items into the $L$ bins respecting the constraint of $k$BP, but whenever required, the algorithm opens a new bin. Then the number of bins required for the algorithm to pack the items in $D_k$ is at most $\max \{ L, (1+ 2 \cdot \epsilon) \cdot OPT(D_k) + k \}$.
\end{lemma}

\begin{proof}
	 Let $I$ be the set of all small items in $D$ 
 and $I_k$ be the $k$ copies of $I$. Let $bins(D_k)$ be the number of bins in the packing of $D_k$. If adding small items do not require a new bin, then $bins(D_k) = L$. Otherwise since  the first item in the last bin cannot be packed to $bins(D_k)-k$ previous bins, each of these bins has less than $(1 - \epsilon) \cdot S$ free size. Thus 
	\begin{align} \label{general:lemma1:equation1}
		\sum {D_k [i]} &\geq (S - \epsilon \cdot S)(bins(D_k)-k) \nonumber \\
		\sum {D_k [i]} &\leq OPT(D_k) \cdot S \nonumber \\
		bins(D_k) &\leq \frac{1}{1 - \epsilon} \cdot OPT(D_k) + k \nonumber \\
		bins(D_k) &\leq (1 + 2 \cdot \epsilon) \cdot  OPT(D_k) + k
	\end{align}
	Therefore,
	\begin{equation} \label{general:lemma1:equation2} bins(D_k) \leq \max \{ L,  (1 + 2 \cdot \epsilon) \cdot OPT(D_k) + k \} 
	\end{equation}
\qed
\end{proof}


Step 2 is done using a linear program based on \emph{configurations}.

\begin{definition}
	A \emph{configuration} (or a \emph{bin type}) is a collection of item sizes which sums to, at most, the bin capacity $S$.
\end{definition}
For example {\cite{enwiki:1139054649}}: suppose there are $7$ items of size $3$, $6$ items of size $4$, and $S=12$. Then, the possible configurations are $ [3,3,3,3], [3,3,3], [3,3], [3], $ $[4,4,4], [4,4], [4],$ $ [3,3,4],[3,4,4], [3,4]$.



Enumerate all possible configurations by the natural numbers from $1$ to $t$. Let $A=\|a_{ij}\|$ be a $m(D) \times t$ matrix, such that for each natural $i\le m(D)$ and $j\le t$ the entry $a_{ij}$ is the number of items of size $c[i]$ in the configuration $j$. Let $\mathbf{n}$ be a $m(D)$-dimensional vector such that for each natural $i\le m(D)$ its $i$th entry is $n[i]$ (the number of items of size $c[i]$). Let $\mathbf{x}$ be a $t$-dimensional vector such that for each natural $j\le t$ we have that $x[j]$ is the number of bins filled with configuration $j$, and $\textbf{1}$ be a $t$-dimensional vector whose each entry is $1$.
Consider the following linear program
\begin{align*}
&& \min  &\quad \mathbf{1 \cdot x}
\\
(C_1)
&&
\text{such that} &\quad A\mathbf{x} = \mathbf{n}
\\
&& &\mathbf{x} \geq 0 
\end{align*}
When $\mathbf{x}$ is restricted to integer entries ($\mathbf{x}\in \mathbb{Z}^t$),
the solution of this linear program defines a feasible bin-packing.
We denote by $F_1$ the fractional relaxation of the above program, where $\mathbf{x}\in \mathbb{R}^t$.

Recall that in $k$BP, each item of $D$ has to appear in $k$ distinct bins. \emph{One can observe that $k$BP uses the same configurations as in the bin-packing, to ensure that each bin contains at most one copy of each item.}
Therefore, the configuration linear program $C_k$ for $k$BP is as follows, 
where $A$, $\mathbf{n}$, and $\mathbf{x}$ are the same as in $C_1$ above
(for $k=1$ it is the same as in \cite{fernandez_de_la_vega_bin_1981}):
\begin{align}
&&	\min  \quad \mathbf{1 \cdot x} \label{eqn:dlvl:4}\\
(C_k)
&&
	\text{such that} \quad A\mathbf{x} &= k\mathbf{n} \label{eqn:dlvl:5} \\
&&	\mathbf{x} &\geq 0 \label{eqn:dlvl:6}
\end{align}

\begin{lemma}
Every integral solution of $C_k$ can be realised as a feasible solution of $k$BP.
\end{lemma}

\begin{proof}
Since each bin can contain at most one copy of each item of $D$, for each natural $j\le t$ if the configuration $j$ can be realized then $a_{ij}\le n[i]$ for each natural $i\le m(D)$. We shall call such configurations \emph{feasible}. We can realise every sequence of feasible configurations as a solution of the $k$BP problem
as follows. 
For each natural $i\le m(D)$, let $d_1,\dots, d_{n[i]}$ be the items from $D$ of size $c[i]$. Let the queue $Q_i$ be arranged of $k$ copies of these items, beginning from  the first copies of $d_1,\dots, d_{n[i]}$ in this order, then of the second copies of these items in the same order, and so forth.
To realize the sequence, consider the first configuration, say, $j$, from the sequence, and for each natural $i\le m(D)$ move $a_{ij}$ items of size $c[i]$ from the queue $Q_i$ into the first bin (or just do nothing when $Q_i$ is already empty), then  similarly proceed the second configuration from the sequence and so forth. Since all configurations are feasible, we never put two copies of the same item in the same bin, so the above procedure constructs a feasible solution to $k$BP.
\qed
\end{proof}

Let $F_k$ be the fractional bin-packing problem corresponding to $C_k$. Step $2$ involves grouping. Grouping reduces the number of different item sizes, and thus reduces the number of constraints and configurations in the fractional linear program $F_k$.

To solve the configuration linear program 
efficiently, both Fernandez de la Vega-Lueker algorithm and Algorithm 1 of Karmarkar Karp use a \emph{linear grouping} technique. In linear grouping, items are divided into groups (of fixed cardinality, except possibly the last group), and each item size (in each group) increases to the maximum item size in that group. 
See appendix \ref{appendix: dlvl-to-kbp} for more detail.

Our extension of the 
Fernandez de la Vega-Lueker and
the Karmarkar-Karp algorithms to $k$BP differs from their original counterparts in mainly two directions. First, in the 
configuration linear program (see the constraint {\ref{eqn:dlvl:5}} in $C_k$), and hence the obtained solution to this configuration linear program is not necessarily the $k$ times copy of the original solution of BP. Second, in greedily adding the small items, see lemma {\ref{general:lemma1}}. In extension of Karmarkar-Karp algorithm 1 to $k$BP we have also shown that getting an integer solution from $\mathbf{x}$ by rounding method may require at most $(k-1)/2$ additional bins. 
We discuss extensions to the Fernandez de la Vega-Lueker and Karmarkar-Karp algorithms and their analyses in subsections \ref{SEC: PTAS: SUBSECTION: dlvl} and \ref{SEC: PTAS: SUBSECTION: kkalgorithms}, respectively.

The inputs to the extension of the algorithms by Fernandez de la Vega-Lueker and Algorithm 1 and Algorithm 2 of Karmarkar-Karp are an input set of items $D$, a natural number $k$, and an 
approximation parameter $\epsilon \in (0,1/2]$. Algorithm 2 of Karmarkar-Karp, in addition, accepts an integer parameter $g > 0$.


\subsection{Fernandez de la Vega-Lueker algorithm to \texorpdfstring{$k$}{k}BP} \label{SEC: PTAS: SUBSECTION: dlvl}
Fernandez de la Vega and Lueker \cite{fernandez_de_la_vega_bin_1981} published a PTAS which, given an input instance $D$ and $\epsilon \in (0,1/2] $, solves a bin-packing problem 
with, at most, $(1+\epsilon) \cdot OPT(D) + 1$ bins. They devised a method called ``adaptive rounding'' for this algorithm. In this method, the given items are put into groups and rounded to the largest item size in that group. This resulting instance will have fewer different item-sizes. This resulting instance can be solved efficiently using a configuration linear program $C_k$ (see appendix \ref{appendix: dlvl-to-kbp} for more detail) .

\paragraph{A high-level description of the extension of Fernandez de la Vega-Lueker algorithm to $k$BP.}
Let $I$ and $J$ be multisets of small and large items in $D$, respectively. After applying linear grouping in $J$, let $U$ be the resulting instance and $C_k$ be the corresponding configuration linear program. An optimal solution to $C_k$ will give us an optimal solution to the corresponding $k$BP instance $U_k$. Ungrouping the items in $U_k$ will give us a solution to $k$BP instance $J_k$. Finally, adding the items in $I_k$, by respecting the constraints of $k$BP, to the solution of $J_k$ (and possibly opening new bins if required) will give us a packing of $D_k$.

Lemma \ref{dlvl:lemma:opt<(1+epsilon)opt} (see appendix \ref{appendix: dlvl-to-kbp} for more detail) bounds the number of bins in an optimal packing of $U_k$. The extension of the algorithm by Fernandez de la Vega-Lueker to $k$BP is given in appendix {\ref{appendix: dlvl-to-kbp}}.

\begin{theorem} \label{dlvl:theorem:bins<=(1+2epsilon)opt+k}
	Generalizing the Fernandez de la Vega-Lueker algorithm to $k$BP will require  $bins(D_k) \leq (1 +  2 \cdot \epsilon ) \cdot OPT(D_k)	 + k$ bins.
\end{theorem}

We give a detailed proof of 
Theorem \ref{dlvl:theorem:bins<=(1+2epsilon)opt+k}, along with the runtime analysis of the Fernandez de la Vega-Lueker algorithm to $k$BP in appendix \ref{appendix: dlvl-to-kbp}.

\subsection{Karmarkar-Karp Algorithms to \texorpdfstring{$k$}{k}BP} \label{SEC: PTAS: SUBSECTION: kkalgorithms}
Karmarkar and Karp \cite{karmarkar-efficient-1982} improved the work done by Fernandez de la Vega and Lueker \cite{fernandez_de_la_vega_bin_1981} mainly in two directions: (1) Solving the linear programming relaxation of $C_1$ 
using a variant of the GLS method \cite{grotschel_ellipsoid_1981} and (2) using a different grouping technique. These improvements led to the development of three algorithms. Their algorithm $3$ is a particular case of the algorithm $2$; we will discuss the generalization of algorithms $1$ and $2$ of Karmarkar-Karp algorithms to solve $k$BP.
Let $LIN(F_k)$ denote the optimal solution to the fractional linear program $F_k$ for $k$BP. We will discuss helpful results relevant to analyzing the generalized version of their algorithms. These results are an extension of the results in \cite{karmarkar-efficient-1982}.
Lemma {\ref{kkalgorithms: opt-dk<=2V(Dk)/s+k}} bounds from above the number of bins needed to pack the items in an optimal packing of some instance $D_k$. Lemma {\ref{kkalgorithms: V(Dk)<=LIN(Dk)S+S(m(Dk)+k)/2}} concerns  obtaining an integer solution from a basic feasible solution of the fractional linear program.
We discuss lemmas \ref{kkalgorithms: opt-dk<=2V(Dk)/s+k} and \ref{kkalgorithms: V(Dk)<=LIN(Dk)S+S(m(Dk)+k)/2} in appendix \ref{kkalgorithms2kBP}.

Before moving further, we would like to mention that if we use some instance (or group) without subscript $k$, we are talking about the instance when $k=1$. 

All Karmarkar-Karp algorithms use a variant of the ellipsoid method to solve the fractional linear program. So, we will talk about adapting this method to $k$BP.

\paragraph{Solving the fractional linear program: } \label{kk:solving-flp}
Solving the fractional linear program \hypertarget{kk:flp}{$F_k$}   involves a variable for each configuration. This results in a large number of variables. The fractional linear program $F_k$ has the following dual $D_F$.
\begin{align} \label{kk:dflp}
	\max  \quad k \cdot \mathbf{n \cdot y} \\
	\text{such that} \quad A^T \mathbf{y} &\leq \mathbf{1} \\
	\mathbf{y} &\geq 0 
\end{align}
The above dual linear program can be solved to any given tolerance $h$ by using a variant of the ellipsoid method that uses an approximate separation oracle \cite{karmarkar-efficient-1982}. The running time of the algorithm is $T(m(D_k),n(D_k)) = O\left(m(D)^8 \cdot  {\ln{m(D)}} \cdot {\ln^2\left(\frac{m(D) \cdot n(D)}{\epsilon \cdot S \cdot h}\right)} + \frac{m(D)^4 \cdot k \cdot n(D) \cdot \ln{m(D)}}{h} \ln{\frac{m(D) \cdot n(D)}{\epsilon \cdot S \cdot h}}\right)$.	We give a description of this variant of the ellipsoid method and its running time in \ref{kkalgorithms: solving-flp}

\paragraph{A high-level description of the extensions of the Karmarkar-Karp algorithms.} We will give a high-level description behind the extension of Karmarkar-Karp algorithms to solve $k$BP. Let $I$ and $J$ be multisets of small and large items in $D$, respectively. Let $U''$ be the instance constructed from $J$ by applying the grouping technique. Construct the configuration linear program $C_k$ for $U''_k$ and solve the corresponding fractional linear program $F_k$. Let $\mathbf{x}$ be the resulting solution. From $\mathbf{x}$, obtain an integral solution for $U''_k$. From this solution, get a solution for $J_k$ by ungrouping the items. Add the items in $I_k$ by respecting the constraints of $k$BP to get a solution for $D_k$.

\subsubsection{Karmarkar-Karp Algorithm 1 extension to \texorpdfstring{$k$}{k}BP:} \label{SUBSUBSECTION: KKAlgorithm1}
Algorithm 1 of the Karmarkar-Karp algorithms uses the linear grouping technique as illustrated in subsection \ref{section:EAA:subsection:general}. We give the extension of the Karmarkar-Karp Algorithm 1 to $k$BP in appendix \ref{kkalgorithms: algorithm1}

\begin{theorem} \label{kkalgorithms:algorithm1:theorem:bin(Dk)<=(1+2kepsilon)OPT+additiveterms}
	Let $bins(D_k)$ denote the 
number of bins produced by Karmarkar-Karp Algorithm 1 extension to $k$BP. Then,
	$bins(D_k) \leq (1 + 2 \cdot k \cdot \epsilon)OPT(D_k) + \frac{1}{2 \cdot \epsilon^2} + (2 \cdot k+1)$.
\end{theorem}

We give the proof of the above Theorem and the running time of the Karmarkar-Karp Algorithm 1 extension to $k$BP in appendix \ref{kkalgorithms: algorithm1}.

\subsubsection{Karmarkar-Karp Algorithm 2 extension to \texorpdfstring{$k$}{k}BP.} \label{SUBSUBSECTION: KKAlgorithm2}
Algorithm 2 of the Karmarkar-Karp algorithms uses the \textit{alternative geometric grouping technique}. Let $J$ be some instance and $g > 1$ be some integer parameter, then, alternative geometric grouping partitions the items in $J$ into groups such that each group contains the necessary number of items so that the size of each group but the last (i.e. the sum of the item sizes in that group) is at least $g \cdot S$. See appendix \ref{kkalgorithms:algorithm2} for more details.

\begin{theorem} \label{kkalgorithm2:theorem:bin(Dk)<=opt+O(klog2opt)}
	Let $bins(D_k)$ denote the 
number of bins produced by Karmarkar-Karp Algorithm 2 extension to $k$BP. Then,
	$bins(D_k) \leq OPT(D_k) + O(k \cdot \log^2 {OPT(D)})$.
\end{theorem}

We give the proof of the above Theorem and the running time of the Karmarkar-Karp Algorithm 2 extension to $k$BP in appendix \ref{kkalgorithms:algorithm2}.


\section{Experiment: \texorpdfstring{$FFk$}{FFk} and \texorpdfstring{$FFDk$}{FFDk} for fair electricity distribution}
\label{section:experiments}
In this section, we describe an experiment checking the performance of the $k$BP adaptations of $FFk$ and $FFDk$ to our motivating application of fair electricity distribution.

\label{experiment:fairelecdistri}

\subsection{Dataset}
\label{experiment:fairelecdistri:dataset}

We use the same dataset of $367$ Nigerian households described in \cite{oluwasuji_solving_2020}.%
\footnote{
	We are grateful to Olabambo Oluwasuji for sharing the dataset with us.
}
This dataset contains the hourly electricity demand for each household for 13 weeks (2184 hours). 
In addition, they estimate for each agent and hour, the \emph{comfort} of that agent, which is an estimation of the utility the agent gets from being connected to electricity at that hour.
For more details about the dataset, readers are encouraged to refer to the papers \cite{oluwasuji_algorithms_2018,oluwasuji_solving_2020}. 
The electricity demand of agents can vary from hour to hour. We execute our algorithms for each hour separately, which gives us essentially 2184 different instances.

As in \cite{oluwasuji_solving_2020}, we use the demand figures in the dataset as mean values; we determine the actual demand of each agent at random from a normal distribution with a standard deviation of $0.05$ (results with a higher standard deviation are presented in appendix \ref{APPENDIX: electricity-distribution-results}).

As in \cite{oluwasuji_solving_2020}, 
we compute the supply capacity $S$ for each day by averaging the hourly estimates of  agents' demand for that day. 
We run nine independent simulations (with different randomization of agents' demands). Thus, the supply changes in accordance with the average daily demand, but cannot satisfy the maximum hourly demand.

\subsection{Experiment}
\label{experiment:fairelecdistri:experiment}

For each hour, we execute the $FFk$ and $FFDk$ algorithms on the households' demands for that hour. We then use the resulting packing to allocate electricity: if the packing returns $q$ bins, then each bin is connected for $1/q$ of an hour, which means that each agent is connected for $k/q$ of an hour.

The authors of \cite{oluwasuji_solving_2020} measure the efficiency and fairness of the resulting allocation, not only by the total time each agent is connected, but also by more complex measures. In particular, they assume that each agent $i$ has a utility function, denoted $u_i$, that determines the utility that the agent receives from being connected to electricity at a given hour.
They consider three different utility models:
\begin{enumerate}
	\item The simplest model is that $u_i$ equals the amount of time the agent $i$ is connected to electricity (this is the model we mentioned in the introduction).
	\item The value $u_i$ can also be equal to the total amount of electricity that the agent $i$ receives. For each hour, the amount of electricity given to $i$ is the amount of time $i$ is connected, times $i$'s demand at that hour.
	\item They also measure the ``comfort'' of the agent $i$ in time $t$ by averaging their demand over the same hour in the past four weeks, and normalizing it by dividing by the maximum value.
\end{enumerate}

For each utility model, they consider three measures
of efficiency and fairness:
\begin{itemize}
	\item Utilitarian: the sum $
	\sum_{i} u_i(x)$ (or the average) of all agents' utilities $u_i$.
	\item Egalitarian: the minimum utility $
	\min_{i} u_i(x)	$ of a single agent,
	\item The maximum difference $\max_{i,j} \{\lvert u_i(x) - u_j(x)\rvert\}$ of utilities between each pair of agents.
\end{itemize}

\subsection{Results}
\label{subsection:implementationresult}
The authors of \cite{oluwasuji_solving_2020} have proposed two models: comfort model (CM) and the supply model. The objective of the CM and the SM model is to maximize the comfort and supply respectively.
For the comparison with the results from \cite{oluwasuji_solving_2020},
we show our results for $FFk$ and $FFDk$ for $k=100$ \footnote{We have checked smaller values of $k$, and found out that the performance increases with $k$. By the time $k$ reached $100$, the performance increase was very slow, so we kept this value.} along with the results in \cite{oluwasuji_solving_2020} in tables \footnote{In tables \ref{results:table:connections-to-supply}-\ref{results:table:comfort-delivered} CM, SM, GA, CSA1, RSA, CSA2 stands for: The Comfort Model, The Supply Model, Grouper Algorithm, Consumption-Sorter Algorithm, Random-Selector Algorithm, Cost-Sorter Algorithm respectively \cite{oluwasuji_solving_2020,oluwasuji_algorithms_2018}} \ref{results:table:connections-to-supply},\ref{results:table:electricity-supplied}, and \ref{results:table:comfort-delivered} . We highlight the best results in bold.
As can be seen in tables \ref{results:table:connections-to-supply}, 
\ref{results:table:electricity-supplied} and 
\ref{results:table:comfort-delivered},
$FFk$ and $FFDk$ outperform previous results in terms of the egalitarian allocation of connection time which is the main objective of this paper. Overall, the comparison of $FFk$ and $FFDk$ with their results are as follows:
\begin{itemize}
    \item[--] $FFk$ and $FFDk$ outperform the previous results in terms of egalitarian allocation of connection time. Maximum Utility difference is also better than all previous results. In terms of utilitarian social welfare metric $FFk$ and $FFDk$ outperforms all previous results except CM.

    \item[--] $FFk$ and $FFDk$ outperform the previous results in terms of utilitarian allocation of supply. In terms of egalitarian social welfare metric and maximum utility difference, $FFk$ and $FFDk$ outperforms all previous results except SM.

    \item[--] $FFk$ and $FFDk$ outperform the previous results in terms of utilitarian allocation of comfort except the CM model. In terms of egalitarian allocation of comfort, $FFk$ and $FFDk$ performs better than previous results except the CM and SM model. Their performance is nearly equivalent to the CM model with better standard deviation and maximum utility difference.
\end{itemize}

In appendix \ref{APPENDIX: electricity-distribution-results}, we have graphs that show how various values of $k$ and varying levels of uncertainty affect the number of connection hours, amount of electricity delivered, and comfort.


\begin{table}[h!]
	\caption{Comparing results of $FFk$ and $FFDk$ for $k=100$ with the results in \cite{oluwasuji_solving_2020} in terms of hours of connection to supply on the average, along with their standard deviation (SD) within parenthesis. In the third column, we have shown the average number of hours an agent is connected to the supply. }
	\label{results:table:connections-to-supply}
	\begin{center}
                    \scalebox{1}{
				\begin{tabular}{|p{0.15\linewidth}|p{0.2\linewidth}|p{0.15\linewidth}|p{0.2\linewidth}|p{0.15\linewidth}|}
					\toprule
					Algorithm & Utilitarian: sum(SD) & Utilitarian: average & Egalitarian(SD) & Maximum Utility Difference \\
					\toprule
					$FFk$ & {716145.3847 (13.9574)} & {1951.3498} & \textbf{1951.3498 (0.0380)} & \textbf{0.0(0.0)} \\
					\toprule
					$FFDk$ & {715891.3137 (13.3774)} & {1950.6575} &  \textbf{1950.6575 (0.0364)} & \textbf{0.0(0.0)} \\
					\toprule
					CM & 717031(3950) & 1953.7629 & 1920(3.24) & 123(2.09) \\
					\toprule
					SM & 709676(3878) & 1933.7221 & 1922(3.41) & 71(2.04) \\
					\toprule
					GA & 629534(4178) & 1715.3515 & 1609(4.69) & 695(3.28) \\
					\toprule
					CSA1 & 647439(3063) & 1764.1389 & 1764(2.27) & 1(0.00) \\
					\toprule
					RSA & 643504(4094) & 1753.4169 & 1753(4.33) & 1(0.00) \\
					\toprule
					CSA2 & 641002(3154) & 1746.5995 & 1746(2.38) & 1(0.00) \\
					\bottomrule
				\end{tabular}}
		\end{center}	
	\end{table}
	
	\begin{table}[h!]
		\caption{Comparing results of $FFk$ and $FFDk$ for $k=100$ with the results in \cite{oluwasuji_solving_2020} in terms of electricity supplied on the average, along with their standard deviation (SD) within parenthesis.}
		\label{results:table:electricity-supplied}
		\begin{center}
                    \scalebox{1}{
				\begin{tabular}{|p{0.15\linewidth}|p{0.25\linewidth}|p{0.25\linewidth}|p{0.25\linewidth}|}
						\toprule
						Algorithm & Utilitarian(SD) & Egalitarian(SD) & Maximum Utility Difference \\
						\toprule
						$FFk$ & \textbf{1364150.4034 (58.6228)} & {0.8067 (0.0001)} & {0.1183 (0.0001)} \\
						\toprule
						$FFDk$ & \textbf{1363494.0885 (48.1455)} & {0.8063 (0.0001)} & {0.1183 (0.0002)} \\
						\toprule
						CM & 1340015(8299) & 0.78(0.01) & 0.17(0.02) \\
						\toprule
						SM & 1347801(8304) & 0.83(0.01) & 0.11(0.02) \\
						\toprule
						GA & 1297020(11264) & 0.35(0.04) & 0.58(0.03) \\
						\toprule
						CSA1 & 1296939(7564) & 0.66(0.02) & 0.28(0.02) \\
						\toprule
						RSA & 1344945(11284) & 0.68(0.03) & 0.25(0.03) \\
						\toprule
						CSA2 & 1345537(7388) & 0.63(0.03) & 0.30(0.02) \\
						\bottomrule
					\end{tabular}}
			\end{center}	
		\end{table}

		\begin{table}[h!]
			\caption{Comparing results of $FFk$ and $FFDk$ for $k=100$ with the results in \cite{oluwasuji_solving_2020} in terms of comfort delivered on the average, along with their standard deviation (SD) within parenthesis.}
			\label{results:table:comfort-delivered}
			\begin{center}
                        \scalebox{1}{
						\begin{tabular}{|p{0.15\linewidth}|p{0.25\linewidth}|p{0.25\linewidth}|p{0.25\linewidth}|}
							\toprule
							Algorithm & Utilitarian(SD) & Egalitarian(SD) & Maximum Utility Difference \\
							\toprule
							$FFk$ & {296630.3583 (7.3626)} & {0.8085 (0.00003)} & {0.1155 (0.00004)} \\
							\toprule
							$FFDk$ & {296493.8100 (7.5569)} & {0.8081 (0.00003)} & {0.1156 (0.00003)} \\
							\toprule
							CM & 303217(3447) & 0.81(0.01) & 0.13(0.02) \\
							\toprule
							SM & 292135(3802) & 0.83(0.01) & 0.09(0.02) \\
							\toprule
							GA & 291021(5198) & 0.38(0.04) & 0.56(0.03) \\
							\toprule
							CSA1 & 291909(3201) & 0.67(0.02) & 0.25(0.02) \\
							\toprule
							RSA & 268564(5106) & 0.65(0.04) & 0.28(0.03) \\
							\toprule
							CSA2 & 270262(3112) & 0.64(0.02) & 0.28(0.02) \\
							\bottomrule
						\end{tabular}}
				\end{center}	
			\end{table}	
%
In appendix {\ref{APPENDIX: electricity-distribution-results}}, we discuss the variation in utilitarian, egalitarian, and maximum-utility difference with different values of $k$ and varying levels of uncertainty (standard deviation). The graphs show that the changes appear to saturate as $k$ increases.

\section{Conclusion and Future Directions} \label{section:conclusion}
We have shown that the existing approximation algorithms, like the First-Fit and the First-Fit Decreasing, can be extended to solve $k$BP. We have proved that, for any $k\geq 1$, the asymptotic approximation ratio for the $FFk$ algorithm is $\left(1.5+\frac{1}{5k}\right) \cdot OPT(D_k) + 3\cdot k $.
We have also proved that the asymptotic approximation ratio for the $NFk$ algorithm is $2$.
We have also demonstrated that the generalization of efficient approximation algorithms like Fernandez de la Vega-Lueker and Karmarkar Karp algorithms solves $k$BP in $(1 + 2 \cdot \epsilon)\cdot OPT(D_k) + k$ and $OPT(D_k) + O(k \cdot \log^2 {OPT(D)})$ bins respectively in polynomial time.
We have also shown the practical efficacy of $FFk$ and $FFDk$ in solving the fair electricity distribution problem. 

Given the usefulness of $k$-times bin-packing to electricity division, an interesting open question is how to determine the optimal value of $k$ --- the $k$ that maximizes the fraction of time each agent is connected --- the fraction $\frac{k}{OPT(D_k)}$. 
Note that this ratio is not necessarily increasing with $k$. For example, consider the demand vector $D = \{11,12,13\}$:
\begin{itemize}
	\item For $k=1$, $OPT(D_k)=2$, so any agent is connected $\frac{1}{2}$ of the time. 
	\item For $k=2$, $OPT(D_k)=3$, so any agent is connected  $\frac{2}{3}$ of the time.
	\item For $k=3$, $OPT(D_k)=5$, so any agent is connected only $\frac{3}{5}<\frac{2}{3}$ of the time. 
\end{itemize}

Some other questions left open are
\begin{enumerate}
	\item 	To bridge the gap in the approximation ratio of $FFk$, between the conjectured lower bound $1.375$ and the upper bound $\left(1.5+\frac{1}{5k}\right) \cdot OPT(D_k) + 3\cdot k $, which converges to $1.5$. 
	\item To prove or disprove that the conjectured bound $\frac{11}{9}OPT + \frac{6}{9}$ is tight for $FFDk$.
\end{enumerate}

\bibliographystyle{splncs04}
\bibliography{kbp-paper}

\newpage
\appendix

\section{\texorpdfstring{$FFk$}{FFk} Algorithm}
\label{appendix:ffk-algorithm}

\subsection{Asymptotic approximation ratio} \label{appendix: ffk-algorithm: ffk-general-case-proof}

In this section we will give the proof of Theorem \ref{ffk:theorem3/2} for the general case. Recall that we have already gave the proof for $k=2$ as a warm up case in Section \ref{section:approx-algo:subsection:approx_algo:ffk}.

\begin{proof}[Proof of Theorem \ref{ffk:theorem3/2} ]
	Recall that we have a basic weighting scheme, in which each item of size $v$ is given a weight 
	\begin{align*}
		w(v) := v/S + r(v),
	\end{align*}
	where $r$ is a \emph{reward function}, computed as follows:
	\[
	r(v) := 
	\begin{cases}
		0 \quad & \text{if } v/S \leq \frac{1}{6}, \\
		\frac{1}{2}(v/S - \frac{1}{6}) \quad & \text{if } v/S \in (\frac{1}{6}, \frac{1}{3}), \\
		1/12 \quad & \text{if } v/S \in [\frac{1}{3}, \frac{1}{2}], \\
		4/12 \quad & \text{if } v/S > \frac{1}{2}.
	\end{cases}
	\]
	
	We modify this scheme later for some items, based on the instances to which they belong.
	
Denote by $u$, the first instance in the $FFk$ packing in which there are underfull bins (if there are no underfull bins at all, we set $u = k$). We modify the item weights as follows:
\begin{itemize}
\item For instances $1,\ldots, u-1$, we reduce the reward of each item of size $v>S/2$ from $4/12$ to $2/12$.
\item For instance $u$, we keep the weights unchanged.
\item For instances $u+1,\ldots, k$, 
we reduce the weights of items as in case $2$ of warm-up case for $k=2$ in the proof of Theorem \ref{ffk:theorem3/2}), that is 
\[
w(v) = 
\begin{cases}
    0 \quad & \text{if } v/S < \frac{1}{3}, \\
    5/12 \quad & \text{if } v/S \in [\frac{1}{3}, \frac{1}{2}], \\
    10/12 \quad & \text{if } v/S > \frac{1}{2}.
\end{cases}
\]
\end{itemize}
Now we analyze the effects of these changes.
\begin{itemize}
\item In the optimal packing, 
the maximum reward of every  bin that contains an item larger than $S/2$ from instances $1,\ldots,u-1$ drops from at most $5/12$ to at most $3/12$, so their weight is at most $15/12$.
Additionally, every bin that contains an item larger than $S/2$ from some instance $j\in\{u+1,\ldots, k\}$, contains only items from instance $j$. Therefore, its weight can be at most $15/12$ (the remaining space in such bin can pack only one item from $[S/3,S/2]$ of positive weight. 
Additionally, the weight of every bin that contains no items larger than $S/2$ remains at most $15/12$.
The only bins that can have a larger weight (up to $17/12$) are bins with items larger than $S/2$ from instance $u$. At most $1/k$ bins in the optimal packing can contain such an item.
Therefore, the average weight of a bin in the optimal packing is at most $\frac{(17/12) + (15/12)\cdot(k-1)}{k}$.
		
		\item 
		In the $FFk$ packing, 
		instances $1,\ldots,u-1$ have no underfull bins. 
		The change in the reward affects only the bins of group 3
		(bins with a single item larger than $S/2$): the weight of each bin in group 3 (which must have one item larger than $2S/3$ since it is not underfull) is at least $2/3+2/12 = 10/12$. 
		Since we only reduce the reward of the item larger than $S/2$ for the instances $1, \ldots ,u-1$; the analysis of bins in groups 4 and 5 is not affected at all by this reduction in reward and therefore leads to their average weight (except excluding one bin per instance in groups 1, 2, and 5) of at least $10/12$.
		
		In instance $u$, the weights are unchanged, so the average bin weight (except excluding one bin per instance in groups 1, 2, and 5) is still at least $10/12$.
		
		Finally, the bins of instances $u+1,\ldots,k$ contain no item of size at most $S/3$ --- since any such item would fit into the underfull bin in instance $u$. Therefore, for each such bin, there are only two options:
		\begin{itemize}
			\item the bin contains one item larger than $S/2$ --- so its weight is $10/12$;
			\item the bin contains two items, each of which is larger than $S/3$ --- so its weight is at least $2\cdot 5/12 = 10/12$.
		\end{itemize}
	\end{itemize}
	Again, the average bin weight (except excluding one bin per instance in groups 1, 2, and 5) of instances $u+1, \ldots, k$ is at least 10/12.
	
	We conclude that the average weight of all bins in the $FFk$ packing (except the $3k$ excluded bins) is at least $10/12$. Therefore
	$$FFk \leq \left(\frac{\frac{(17/12) + (15/12)\cdot(k-1)}{k}}{10/12} \cdot OPT\right) + 3k=\left(1.5+\frac 1{5k}\right) \cdot OPT + 3k.$$	
 \qed
\end{proof}

\subsection{Worst-case example} \label{appendix: ffk-worst-case-example-johnsonpaper}
\begin{lemma} \label{lower bound 1}
    The approximation ratio of $FFk(D_k)$ for $k=2$ for the worst case examples given in \cite{Dósa_Sgall_2014,Johnson1974} is $1.35$. For the example given in \cite{Johnson1974}, the approximation ratio continues to decrease as $k$ increases.
\end{lemma}
\begin{proof}
    Dosa and Sgall \cite{Dósa_Sgall_2014} gave a simple example for the lower bound construction of $FF$ as following: For a very small $\epsilon$, first there are $10n$ \emph{small} items of size approximately $1/6$ (smaller and bigger), then comes $10n$ \emph{medium} items of size approximately $1/2$ (smaller and bigger), and finally $10n$ \emph{large} items of size $1/2 + \epsilon$ follows. They have proved that for this list $D$, $OPT(D)=10n$ and $FF(D)=17n$. Now, consider the packing of the items in $D_2$. After packing the first instance of $D$, $FFk$ packs the \emph{small} and \emph{medium} items into the existing bins. Only each of the \emph{large} items require a separate bin to pack. So it is easy to see that $FFk(D_2)=(17+10)n$ while $OPT(D_2)=(10+10)n$. This will give us $\frac{FFk(D_2)}{OPT(D_2)} = \frac{27}{20} = 1.35$. Note that here $OPT$ is arbitrarily big.
 
Johnson, Demers, Ullman, Garey and Graham \cite{Johnson1974} provide the following example for $FF$. We now show that, on the same example, $FFk$ achieves an approximation ratio of $1.35$.

All the item sizes in this example 
$D$ are in non-decreasing order and are as follows
$D = \{6 (7),10 (7),16 (3),34 (10),51 (10)\}$ and bin size $S = 101$.
The number in the parenthesis denotes the occurrence of that item in $D$. In the above example $6 (7)$ expands to $\{6,6,6,6,6,6,6\}$. Before proceeding further, we describe some notation that we will use. For two bins $B_i, B_j$, $i<j$ implies bin $B_i$ has been created before bin $B_j$.


\begin{figure}
\centering
    \begin{subfigure}[b]{0.3\textheight}
        \includegraphics[height=0.9\textheight, width=0.9\textwidth, keepaspectratio]{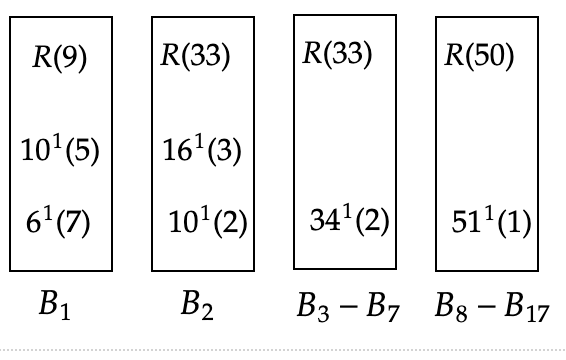}
        \caption{$FFk$ packing for $k=1$.}
        \label{ffk_lb_k1}
    \end{subfigure}

    \begin{subfigure}[b]{0.55\textheight}
        \includegraphics[height=0.9\textheight, width=0.9\textwidth, keepaspectratio]{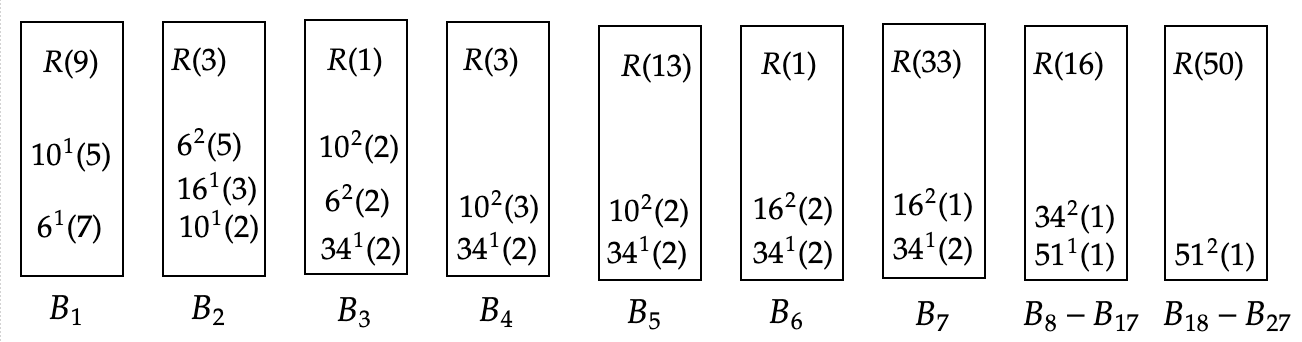}
        \caption{$FFk$ packing for $k=2$.}
        \label{ffk_lb_k2}
    \end{subfigure}

    \begin{subfigure}[b]{0.4\textheight}
        \includegraphics[height=0.9\textheight, width=0.9\textwidth, keepaspectratio]{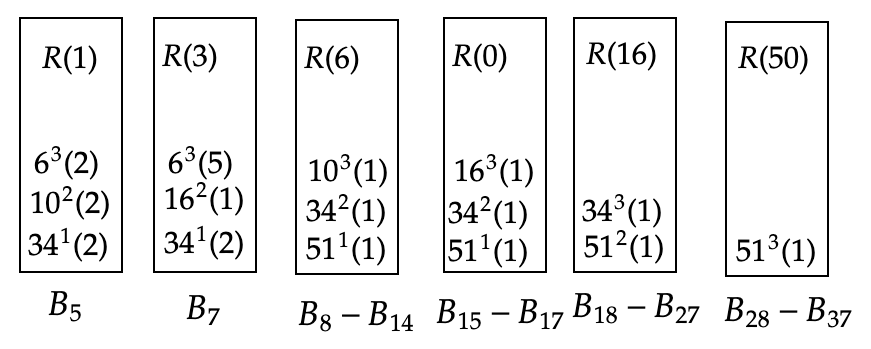}
        \caption{$FFk$ packing for $k=3$. Note that bins from $FFk(D_2)$ packing in which no item from future instances can be packed are not shown. }
        \label{ffk_lb_k3}
    \end{subfigure}

    \begin{subfigure}[b]{0.4\textheight}
        \includegraphics[height=0.9\textheight, width=0.9\textwidth, keepaspectratio]{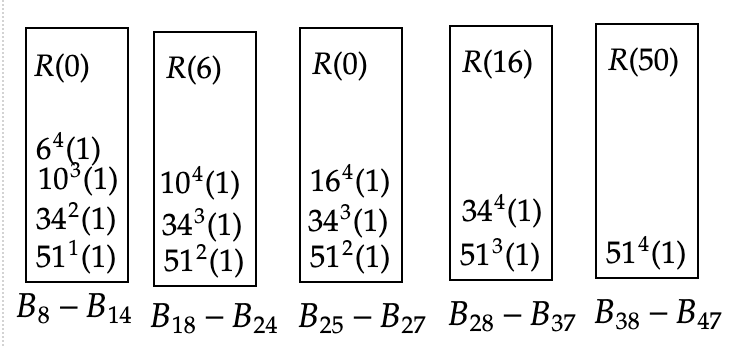}
        \caption{$FFk$ packing for $k=4$. Note that bins from $FFk(D_3)$ packing in which no item from future instances can be packed are not shown.}
        \label{ffk_lb_k4}
    \end{subfigure}

    \begin{subfigure}[b]{0.4\textheight}
        \centering
        \includegraphics[height=0.9\textheight, width=0.9\textwidth, keepaspectratio]{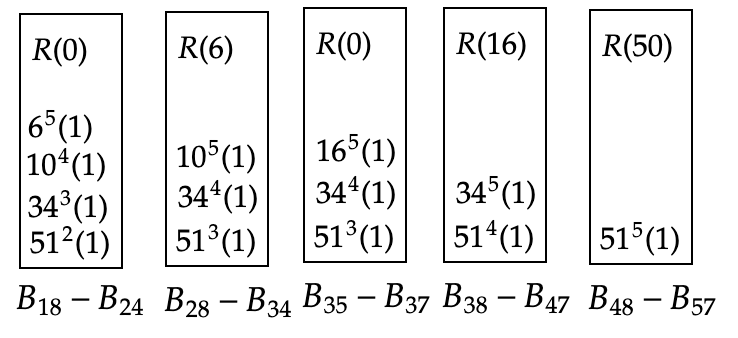}
        \caption{$FFk$ packing for $k=5$. Note that bins from $FFk(D_4)$ packing in which no item from future instances can be packed are not shown.}
        \label{ffk_lb_k5}
    \end{subfigure}
    \caption{Figures \ref{ffk_lb_k1}-\ref{ffk_lb_k5} show the $FFk$ packing for $D$ for different values of $k$. For the lack of space we are not showing the bins in the packing of $D_{k-1}$ in which no item from future instances can be packed. Note the pattern in packing in figures \ref{ffk_lb_k4} and \ref{ffk_lb_k5}.}
    \label{ffk lower bound example}
\end{figure}

Note that $OPT$ for $k=1$ can result the following bin-packing 

$O_1 = O_2 = O_3 = [51,34,16] $,

$O_4 = O_5 = O_6 = O_7 = O_8 =O_9 = O_{10} = [51,34,10,6]$.

It is easy to observe that optimal number of bins for packing $D_k$ is $OPT(D_k)=k \cdot OPT(D)$.

In figure {\ref{ffk lower bound example}} we have shown the packing of $D_1, D_2, D_3, D_4$, and $D_5$. Each box in the figure represents a bin. Each bin represents the items packed into the bin along with their count and the instance from which they belong. Each bin also represents the space remaining in the bin. For example, in figure {\ref{ffk_lb_k1}}, $6^1(7)$, inside bin $B_1$ represents that there are 7 items of size 6 from first instance of $D$. $R(9)$ in bin $B_1$ denotes the remaining space in that bin.

Figures {\ref{ffk_lb_k1}} and {\ref{ffk_lb_k2}} show the $FFk$ packing of the items for instance $D_1$ and $D_2$ respectively. 
Note that there are some bins in which the items from future bins cannot be packed (either there is no space or due to the violation of the $k$BP constraint). For the lack of space we do not show such bins. For example, 
in figure {\ref{ffk_lb_k3}} we have the shown the $FFk$ packing of the items for instance $D_3$. $FFk$ packs the items in the third instance of $D$ after packing the items in the first two instances. However, after packing the second instance, bins $B_1, B_2, B_3, B_4$ cannot pack any of the items from the third instance (in fact any of the items from all the future instances). Therefore, we do show these bins in figure {\ref{ffk_lb_k3}}. 
Note that there is a pattern in the $FFk$ packing of $D_4$ and $D_5$. It is easy to see that this pattern repeats itself in the packing of $D_6, D_7, \ldots, D_k$. Also note that while packing the items in $D_k$ we add 10 new bins in addition to the bins in the $FFk$ packing of $D_{k-1}$. Therefore, 
$\frac{FFk(D_k)}{OPT(D_k)} = \frac{17 + 10(k-1)}{10k} \leq 1.35$ for $k>1$. One can observe that as $k$ increases, the approximation ratio continues to decrease.
\qed
\end{proof}

\begin{lemma} \label{lower bound for ffk}
    The lower bound of $FFk(D_k)$ for $k=2$ is 1.375.
\end{lemma}
\begin{proof}
    Consider the input $D= \{371, 659, 113, 47, 485, 3, 228, 419,468,581,626\}$ and $S=1000$.
For $k=1$, $FFk$ will result in the following packing: \\
$B_1 = [371^1, 113^1, 47^1, 3^1, 228^1], sum=762$, \\
$B_2 = [659^1], sum=659 $, \\
$B_3 = [485^1, 419^1], sum = 904$, \\
$B_4 = [468^1], sum = 468$, \\
$B_5= [581^1], sum = 581$, \\
$B_6 = [626^1], sum = 626$ \\
We use the word ``SAME" for the content of a bin if a bin cannot pack any of the items from the future instances of $D$.
This happens because every item in a future instance is either too large to fit into that bin, or has the same type as an item already packed into it. 
Note that bin $B_1$ is such bin. 
For $k=2$, $FFk$ will result in the following packing: \\
$B_1 = \text{SAME}$,\\
$B_2 = [659^1, 113^2, 47^2, 3^2], sum=822$, \\
$B_3 = [485^1, 419^1], sum=904$, \\
$B_4 = [468^1, 371^2], sum=839$, \\
$B_5= [581^1, 228^2], sum=809$, \\
$B_6 = [626^1], sum=626$, \\
$B_7 = [659^2], sum=659$, \\
$B_8 = [485^2, 419^2], sum=904$, \\
$B_9 = [468^2], sum=468$, \\
$B_{10} = [581^2], sum=581$, \\
$B_{11} = [626^2], sum=626$, \\

The items in $D$ can be optimally  packed in bins $[626,371,3], [659,228,113], [419,581], [468,485,47]$. The size of each of these bins is exactly $1000$. Therefore, the optimal packing of $D_k$ requires $k \cdot OPT(D) = 4k$ bins. 
Therefore, $\frac{FFk(D_k)}{OPT(D_k} = \frac{11}{8} = 1.375$ for $k = 2$. 
\qed
\end{proof}

\begin{remark}
    The approximation ratio for the example instance given in lemma \ref{lower bound for ffk} continues to decrease as $k$ increases.
\end{remark}
\begin{proof}
    The packing for $k=2$ is given in \ref{lower bound for ffk}. We will continue from this packing.
Note that bin $B_2$ cannot pack any of the items from the future instances of $D$.

For $k=3$, $FFk$ will result in the following packing: \\
$B_1 = \text{SAME}, B_2 = \text{SAME}$\\
$B_3 = [485^1, 419^1, 47^3, 3^3], sum=954$, \\
$B_4 = [468^1, 371^2, 113^3], sum=952$, \\
$B_5= [581^1, 228^2], sum=809$, \\
$B_6 = [626^1, 371^3], sum=997$, \\
$B_7 = [659^2, 228^3], sum=887$, \\
$B_8 = [485^2, 419^2], sum=904$, \\
$B_9 = [468^2, 485^3], sum=953$, \\
$B_{10} = [581^2, 419^3], sum=1000$, \\
$B_{11} = [626^2], sum=626$, \\
$B_{12} = [659^3], sum=659$, \\
$B_{13} = [468^3], sum=468$, \\
$B_{14} = [581^3], sum=581$, \\
$B_{15} = [626^3], sum=626$ \\
Note the bin $B_3$ and $B_{10}$ cannot pack any of the items from the future instances of $D$.

For $k=4$, $FFk$ will result in the following packing: \\
$B_1 = \text{SAME}, B_2 = \text{SAME}, B_3 = \text{SAME}, B_{10} = \text{SAME}$\\
$B_4 = [468^1, 371^2, 113^3, 47^4], sum=999$, \\
$B_5= [581^1, 228^2, 113^4, 3^4], sum=925$, \\
$B_6 = [626^1, 371^3], sum=997$, \\
$B_7 = [659^2, 228^3], sum=887$, \\
$B_8 = [485^2, 419^2], sum=904$, \\
$B_9 = [468^2, 485^3], sum=953$, \\
$B_{11} = [626^2, 371^4], sum=997$, \\
$B_{12} = [659^3, 228^4], sum=887$, \\
$B_{13} = [468^3, 485^4], sum=953$, \\
$B_{14} = [581^3, 419^4], sum=1000$, \\
$B_{15} = [626^3], sum=626$, \\
$B_{16} = [659^4], sum=659$, \\
$B_{17} = [468^4], sum=468$, \\
$B_{18} = [581^4], sum=581$, \\
$B_{19} = [626^4], sum=626$ \\
Note that bin $B_4$ and $B_{14}$ cannot pack any of the items from the future instances of $D$.

For $k=5$, $FFk$ will result in the following packing: \\
$B_1 = \text{SAME}, B_2 = \text{SAME}, B_3 = \text{SAME}, B_4 = \text{SAME}, B_{10} = \text{SAME}, B_{14}=\text{SAME}$\\
$B_5= [581^1, 228^2, 113^4, 3^4, 47^5], sum=972$, \\
$B_6 = [626^1, 371^3, 3^5], sum=1000$, \\
$B_7 = [659^2, 228^3, 113^5], sum=1000$, \\
$B_8 = [485^2, 419^2], sum=904$, \\
$B_9 = [468^2, 485^3], sum=953$, \\
$B_{11} = [626^2, 371^4], sum=997$, \\
$B_{12} = [659^3, 228^4], sum=887$, \\
$B_{13} = [468^3, 485^4], sum=953$, \\
$B_{15} = [626^3, 371^5], sum=997$, \\
$B_{16} = [659^4, 228^5], sum=887$, \\
$B_{17} = [468^4, 485^5], sum=953$, \\
$B_{18} = [581^4, 419^5], sum=1000$, \\
$B_{19} = [626^4], sum=626$, \\
$B_{20} = [659^5], sum=659$, \\
$B_{21} = [468^5], sum=468$, \\
$B_{22} = [581^5], sum=581$, \\
$B_{23} = [626^5], sum=626$ \\
Note that bin $B_5, B_6, B_7 B_{18}$ cannot pack any of the items from the future instances of $D$.

For $k=6$, $FFk$ will result in the following packing: \\
$B_1 = \text{SAME}, B_2 = \text{SAME}, B_3 = \text{SAME}, B_4 = \text{SAME}, B_{5} = \text{SAME}, B_{6} = \text{SAME}, B_{7} = \text{SAME}, B_{10} = \text{SAME}, B_{14}=\text{SAME}, B_{18} = \text{SAME}$\\
$B_8 = [485^2, 419^2, 47^6, 3^6], sum=954$, \\
$B_9 = [468^2, 485^3], sum=953$, \\
$B_{11} = [626^2, 371^4], sum=997$, \\
$B_{12} = [659^3, 228^4, 113^6], sum=1000$, \\
$B_{13} = [468^3, 485^4], sum=953$, \\
$B_{15} = [626^3, 371^5], sum=997$, \\
$B_{16} = [659^4, 228^5], sum=887$, \\
$B_{17} = [468^4, 485^5], sum=953$, \\
$B_{19} = [626^4, 371^6], sum=997$, \\
$B_{20} = [659^5, 228^6], sum=887$, \\
$B_{21} = [468^5, 485^6], sum=953$, \\
$B_{22} = [581^5, 419^6], sum=1000$, \\
$B_{23} = [626^5], sum=626$, \\
$B_{24} = [659^6], sum=659$, \\
$B_{25} = [468^6], sum=468$, \\
$B_{26} = [581^6], sum=581$, \\
$B_{27} = [626^6], sum=626$ \\
Note that bins $B_8, B_{12}, B_{22}$ cannot pack any of the items from the future instances of $D$.

For $k=7$, $FFk$ will result in the following packing: \\
$B_1 = \text{SAME}, B_2 = \text{SAME}, B_3 = \text{SAME}, B_4 = \text{SAME}, B_{5} = \text{SAME}, B_{6} = \text{SAME}, B_{7} = \text{SAME}, B_{8} = \text{SAME}, B_{10} = \text{SAME}, B_{12} = \text{SAME}, B_{14}=\text{SAME}, B_{18} = \text{SAME}, B_{22} = \text{SAME},$\\
$B_9 = [468^2, 485^3, 47^7], sum=1000$, \\
$B_{11} = [626^2, 371^4, 3^7], sum=1000$, \\
$B_{13} = [468^3, 485^4], sum=953$, \\
$B_{15} = [626^3, 371^5], sum=997$, \\
$B_{16} = [659^4, 228^5, 113^7], sum=1000$, \\
$B_{17} = [468^4, 485^5], sum=953$, \\
$B_{19} = [626^4, 371^6], sum=997$, \\
$B_{20} = [659^5, 228^6], sum=887$, \\
$B_{21} = [468^5, 485^6], sum=953$, \\
$B_{23} = [626^5, 371^7], sum=997$, \\
$B_{24} = [659^6, 228^7], sum=887$, \\
$B_{25} = [468^6, 485^7], sum=953$, \\
$B_{26} = [581^6, 419^7], sum=1000$, \\
$B_{27} = [626^6], sum=626$, \\
$B_{28} = [659^7], sum=659$, \\
$B_{29} = [468^7], sum=468$, \\
$B_{30} = [581^7], sum=581$, \\
$B_{31} = [626^7], sum=626$ \\
Note that bins $B_9, B_{11}, B_{16}, B_{26}$ cannot pack any of the items from the future instances of $D$.

For $k=8$, $FFk$ will result in the following packing: \\
$B_1 = \text{SAME}, B_2 = \text{SAME}, B_3 = \text{SAME}, B_4 = \text{SAME}, B_{5} = \text{SAME}, B_{6} = \text{SAME}, B_{7} = \text{SAME}, B_{8} = \text{SAME}, B_{9} = \text{SAME}, B_{10} = \text{SAME}, B_{11} = \text{SAME}, B_{12} = \text{SAME}, B_{14}=\text{SAME}, B_{16} = \text{SAME}, B_{18} = \text{SAME}, B_{22} = \text{SAME}, B_{26} = \text{SAME},$\\
$B_{13} = [468^3, 485^4, 47^8], sum=1000$, \\
$B_{15} = [626^3, 371^5, 3^8], sum=1000$, \\
$B_{17} = [468^4, 485^5], sum=953$, \\
$B_{19} = [626^4, 371^6], sum=997$, \\
$B_{20} = [659^5, 228^6, 113^8], sum=1000$, \\
$B_{21} = [468^5, 485^6], sum=953$, \\
$B_{23} = [626^5, 371^7], sum=997$, \\
$B_{24} = [659^6, 228^7], sum=887$, \\
$B_{25} = [468^6, 485^7], sum=953$, \\
$B_{27} = [626^6, 371^8], sum=997$, \\
$B_{28} = [659^7, 228^8], sum=887$, \\
$B_{29} = [468^7, 485^8], sum=953$, \\
$B_{30} = [581^7, 419^8], sum=1000$, \\
$B_{31} = [626^7], sum=626$, \\
$B_{32} = [659^8], sum=659$, \\
$B_{33} = [468^8], sum=468$, \\
$B_{34} = [581^8], sum=581$, \\
$B_{35} = [626^8], sum=626$ \\
Note that bins $B_{13}, B_{15}, B_{20}, B_{30}$ cannot pack any of the items from the future instances of $D$.

At this point, we note a repeating pattern:

\begin{itemize}
    \item Item $371^7, 371^8$ has been packed in bins $B_{23}, B_{27}$ respectively. These bins have the same content $[626, 371]$ (without superscript). For $k=9$ it will pack $371^9$ in bin $B_{31}$, and afterwards this bin will also have the same content $[626, 371]$. 
    \item 
    Items $659^7, 659^8$ have been packed in bins $B_{28}, B_{32}$ respectively.
    Item $659^9$ cannot be packed into any of the previous bins and hence will require a new bin $B_{36}$ to pack. 
    \item 
    Items $113^7, 113^8$ have been packed in bins $B_{16}, B_{20}$ respectively. These bins have the same content $[659, 228, 113]$ (without superscript). 
    For $k=9$ it will pack $113^9$ in bin $B_{24}$ and after packing this bin will also have the same content $[659, 228, 113]$.
    \item Items $47^7, 47^8$ have been packed in bins $B_{9}, B_{13}$ respectively. These bins have the same content $[468, 485, 47]$ (without superscript). For $k=9$ it will pack $47^9$ in bin $B_{17}$ and after packing this bin will also have the same content $[468, 485, 47]$.
    \item 
    Items $485^7, 485^8$ have been packed in bins $B_{25}, B_{29}$ respectively. These bins have the same content $[468, 485]$ (without superscript). For $k=9$ it will pack $485^9$ in bin $B_{33}$ and afterwards this bin will also have the same content $[468, 485]$.
    \item
    Items $3^7, 3^8$ have been packed in bins $B_{11}, B_{15}$ respectively. These bins have the same content $[626, 371, 3]$ (without superscript). For $k=9$ it will pack $3^9$ in bin $B_{19}$ and after packing this bin will also have the same content $[626, 371, 3]$.
    \item 
Items $228^7, 228^8$ have been packed in bins $B_{24}, B_{28}$ respectively. These bins have the same content $[659, 228]$ (without superscript). For $k=9$ it will pack $228^9$ in bin $B_{32}$ and after packing this bin will also have the same content $[659, 228]$.
\item 
Items $419^7, 419^8$ have been packed in bins $B_{26}, B_{30}$ respectively. These bins have the same content $[581, 419]$ (without superscript). For $k=9$ it will pack  $419^9$ in bin $B_{34}$ and afterwards this bin will also have the same content $[581, 419]$.
\item 
Items $468^9, 581^9, 626^9$ cannot be packed into any of the previous bins and hence will require new bins $B_{37}, B_{38}, B_{39}$ respectively. 
Note that for $k=8$ these items were packed in new bins $B_{29}, B_{30}, B_{31}$ and for $k=7$ these items were packed in new bins $B_{25}, B_{26}, B_{27}$.  
\end{itemize}
One can observe that the content of the bins (except the $\text{SAME}$ bins) remains same for $k=7,8,9$. This pattern continues for further values of $k$. After $k \geq 2$, packing of $D_k$ requires 4 new bins than the packing of $D_{k-1}$. 
Therefore, $bins(D_k) = 11 + 4\cdot(k-2) = 4\cdot k + 3$.

The items in $D$ can be optimally  packed in bins $[626,371,3], [659,228,113], [419,581], [468,485,47]$. The size of each of these bins is exactly $1000$. Therefore, the optimal packing of $D_k$ requires $k \cdot OPT(D) = 4k$ bins. 
Therefore, $\frac{FFk(D_k)}{OPT(D_k} = \frac{11 + 4 \cdot (k-2)}{4 \cdot k} = 1 + \frac{3}{4k} \leq 1.375$ for $k \geq 2$. One can observe that as $k$ increases, the approximation ratio continues to decrease.
\qed
\end{proof}


\section{\texorpdfstring{$FFDk$}{FFDk} algorithm} \label{appendix:FFDk-algorithm}
\paragraph{Challenges in extending existing proof:} Existing proofs for the $FFD$ are based on the assumption that the last bin contains a single item, and no other item(s) are packed after that, i.e. the smallest item is the only item in the last bin of $FFD$ packing. We cannot say the same in the case of $k$BP when $k>1$. For example let $D= \{ 103,102,101\}$ and ${S}=205$. Then $FFDk$ packing when $k=1$ is $\{103,102\}, \{101\}$  whereas for $k=2$ the packing is $\{103, 102\}, \{101,103\}, \{102,101\}$. Hence, the existing proofs for $FFD$ cannot be extended to $FFDk$.

\section{\texorpdfstring{$NFk$}{NFk} algorithm} \label{appendix:NFk-algorithm}

\begin{proof}[Proof of Theorem \ref*{NFk:theorem:asymptotic-ratio-2}]
	The idea behind the proof is due to \cite{Johnson1973}. Let the number of bins in the ${NFk}$ packing of $D_k$ is ${NFk}(D_k)$. Let these bins are ordered in the sequence in which they are opened. Then, for any two bins $B_{i-1}$ and $B_{i}$ , where $i$ is not the index of the first bin, it is true that:
	\[
	V(B_{i-1}) + V(B_i) > S
	\]
	Therefore,
	\[
	V(D_k)=V(B_1) + \ldots + 	V(B_{{NFk}(D_k)}) > \left\lfloor \frac{{NFk}(D_k)}{2} \right\rfloor \cdot S
	\]
	
	Since $OPT(D_k) \geq V(D_k)/S$,
	\[
	{NFk}(D_k) \leq 2 \cdot OPT(D_k) + 1 
	\]
	For the lower bound we consider the example given in \cite{Zheng_Luo_Zhang_2015}. Let the bin capacity is 1. Let $D$ be the input instance in which there are $y$ (for some large $y$) copies of the two items having sizes $\{1/2, \epsilon\}$ where $0 < \epsilon < 1/y$. ${NFk}$ will pack $D_k$ in $k\cdot y$ bins while optimal packing of $D_k$ consists of $k \cdot (y/2+1)$ bins. This will give us an asymptotic ratio of $2$.
 \qed
\end{proof}


\section{Fernandez de la Vega-Lueker algorithm to \texorpdfstring{$k$}{k}BP}
\label{appendix: dlvl-to-kbp}
Configuration linear program $C_k$ can be solved as follows:

\textbf{Using exhaustive search:} Since there are $n(D)$  items, a configuration can repeat at most $n(D)$  times. Therefore, $(x_\tau)_{\tau \in T} \in \{0,1, \ldots,n(D)\}^{\lvert T \rvert}$ 
where $T$ is the set of all possible configurations, $\tau$ is some configuration in $T$, and $x_\tau$ is the number of bins filled with the configuration $j$.
Each configuration contains at most $1/\epsilon$ items, and there are $m(D)$ different item sizes. Therefore, the number of possible configurations is at most $m(D)^{1/\epsilon}$. Now, we can check if enough slots are available for each size $c[i]$. Finally, output the solution with the minimum number of bins. Running time then is $m(D) \times n(D)^{m(D)^{1/\epsilon}}$, which is polynomial in $n(D)$ when $m(D)$ and $\epsilon$ are fixed.

\paragraph{Linear Grouping: } Let $D$ be some instance of the bin-packing problem
and $g > 1$ be some integer parameter. Order the items in $D$ in a non-increasing order. Let $U$ be the instance obtained by making groups of the items in $D$ of cardinality $g$ and then rounding up the items in each group by the maximum item size in that group. Let $U'$ be the group of the $g$ largest items and $U''$ be the instance consisting of groups from the second to the last group in $U$. Then 
$OPT(U'') \leq OPT(D)$ \cite{karmarkar-efficient-1982} and 
$OPT(U') \leq g$, because each item in $U'$ can be packed into a single bin. So
$$OPT(D) \leq OPT(U'' \cup U') \leq OPT(U'') + OPT(U') \leq OPT(U'') + g.$$
It implies the below Lemma \ref{general:lemma2} due to \cite{karmarkar-efficient-1982},
where 
 $LIN(D)$ denotes the value of the fractional bin-packing problem $F_1$ associated with the instance $D$:
\begin{lemma} \label{general:lemma2}
    1. $OPT(U'') \leq OPT(D) \leq OPT(U'') + g$. $\quad$
	2. $LIN(U'') \leq LIN(D) \leq LIN(U'') + g$. $\quad$
	3. $V(U'') \leq V(D) \leq V(U'') + g$.
\end{lemma}

If the items in $D$ are arranged in a non-decreasing order, then $U''$ will be the instance consisting of groups from first to the second from the last group, and $U'$  be the last group consisting of the $\leq g$ large items. The above same results hold in this case as well.

We extend the algorithm by de la Vega and Lueker to solve $k$BP as follows:

\begin{algorithm}[H]
	\DontPrintSemicolon
	\KwIn{A set $D$ of items, an integer $k$, and $\epsilon \in (0, 1/2]$.}
	\KwOut{A bin-packing of $D_k$.}
	Let $I$ and $J$ be multisets of small (of size $\le\epsilon \cdot S$) and large (of size $>\epsilon \cdot S$) items in $D$, respectively.  \;
	Sort the items in $J$ in a non-decreasing order of their sizes. \;
	Construct an instance $U$ from $J$ by 
 the linear grouping with $g = n(J) \cdot {\epsilon}^2$ and round up the sizes in each group to the maximum size in that group. \;
	Optimally solve the configuration linear program $C_k$ corresponding to the rounded problem $U$. 
	This will give us an optimal solution to the $k$BP instance $U_k$.  \;
	 To get a solution for $J_k$, ``ungroup'' the items, that is replace the items in groups with the original items in that group, as in step $3$, prior to rounding up.  \;
	Greedily add small items in $I_k$ by respecting constraint of $k$BP to get a solution for $D_k$. \;
	\caption{Fernandez de la Vega-Lueker Algorithm to $k$BP}
	\label{delaVegaLueker}
\end{algorithm}	

The instances $I_k$ and $J_k$ consist of $k$ copies of instances $I$ and $J$, respectively.

\begin{lemma} \label{dlvl:lemma:opt<(1+epsilon)opt}
	Let  $OPT(U_k)$ denote the optimal number of bins in solving $C_k$. Then, $OPT(U_k) < (1 + \epsilon) \cdot OPT(J_k) $.
\end{lemma}
\begin{proof}
	Let $L$ be the problem resulting from rounding down the items in each group but the first by the maximum of the previous group. Then, the instance $L_k$ consists of $k$ copies of instance $L$. The items in the first group are  rounded down to $0$. Since bin-packing  is monotone, 
	\begin{align} \label{equation:dlvlp:1}
		OPT(L_k) \leq OPT(J_k) \leq OPT(U_k)
	\end{align} 
	Note that $L$ and $U$ differ by a group of at most 
 $n(J) \cdot {\epsilon}^2$ items. Therefore, the difference between $L_k$ and $U_k$ is, at most, $k \cdot n(J) \cdot {\epsilon}^2$. Hence,
	\begin{align} \label{equation:dlvlp:2}
		OPT(U_k) - OPT(L_k) \leq k\cdot n(J)\cdot \epsilon^2 \nonumber \\
		OPT(U_k) \leq OPT(J_k) + k\cdot n(J)\cdot \epsilon^2
	\end{align}
	Since all items in $U_k$ 
 (and hence in $J_k$) have a size larger than $\epsilon \cdot S$, the number of items in a bin is at most $1/\epsilon$.  Hence, the minimum number of bins required to pack 
 the items in $J_k$
 are at least
 $k \cdot n(J) \cdot \epsilon$. Hence,
	\begin{equation} \label{equation:dlvlp:3}
		k \cdot n(J) \cdot \epsilon^2 \leq \epsilon \cdot OPT(J_k).
	\end{equation}
	By \ref{equation:dlvlp:2} and \ref{equation:dlvlp:3},  
	\begin{equation} \label{dlvlres1}
		OPT(U_k) \leq (1 + \epsilon)OPT(J_k)
	\end{equation}	
 \qed
\end{proof}


\begin{corollary} \label{dlvl:corollary:bins<(1+epsilon)opt}
	$ bins(J_k) \leq (1+\epsilon) \cdot OPT(J_k)$.
\end{corollary}
\begin{proof}
	From steps $3$ and $5$ of the algorithm, it is clear that from a packing of $U_k$ we can obtain a packing of $J_k$ with the same number of bins. Therefore, $bins(J_k) \leq (1+\epsilon) \cdot OPT(J_k)$.
 \qed
\end{proof}


\begin{proof}[of Theorem \ref{dlvl:theorem:bins<=(1+2epsilon)opt+k}]
	Let $I$ be the set of all items of size at most $\epsilon \cdot S$ in $D$, and $I_k$ be $k$ copies of the instance $I$. In the step $6$, while packing the items in $I_k$ if fewer than $k$ new bins are opened, then from corollary $\ref{dlvl:corollary:bins<(1+epsilon)opt}$ it is clear that $bins(D_k) \leq bins(J_k) + k \leq (1 + \epsilon) \cdot OPT(J_k) + k$. Since bin-packing is monotone, $OPT(J_k) \leq OPT(D_k)$. Therefore, $bins(D_k) \leq (1 + \epsilon) \cdot OPT(D_k) + k$.
	
	Now assume that at least $k$ new bins are opened at the step $6$ of the algorithm. Then, 
 From Lemma \ref{general:lemma1},
	\begin{align} \label{dlvl:theorem1:equation1}
		bins(D_k) &\leq \max \{(1+\epsilon) \cdot OPT(D_k) + k, (1+2 \cdot \epsilon) \cdot OPT(D_k) + k \} \nonumber \\
		&\leq (1+2 \cdot \epsilon) \cdot OPT(D_k) + k
	\end{align}
\qed
\end{proof}

\paragraph{Running time of algorithm \ref{delaVegaLueker}: }
\label{dlvl:algorithmruntimeanalysis} Time $O(n(D)\log{n(D)})$ is sufficient for all the steps except the step 4 in algorithm \ref{delaVegaLueker}. Step 4 takes time $O\left(m(D) \cdot {n(D)} ^ {{m(D)} ^ {1/\epsilon}}\right)$. Therefore, the time taken by the algorithm is $O\left(m(D) \cdot {n(D)} ^ {{m(D)} ^ {1/\epsilon}}\right)$, which is polynomial in $n(D)$ given $m(D)$ and $\epsilon$.



\section{Karmarkar-Karp algorithms to \texorpdfstring{$k$}{k}BP} \label{kkalgorithms2kBP}
\begin{lemma} 
\label{kkalgorithms: opt-dk<=2V(Dk)/s+k}
	$OPT(D_k) \leq 2 \cdot V(D_k)/S + k$ 
\end{lemma}
\begin{proof}
	 When $k=1$ then from \mbox{\cite{karmarkar-efficient-1982}} we know that $OPT(D) \leq \frac{2}{S} \cdot V(D)+ 1$. Therefore $$OPT(D_k) \leq k \cdot OPT(D) \leq \frac{2}{S}\cdot k \cdot V(D)+ k = \frac{2}{S} \cdot V(D_k)+ k.$$
    \qed
\end{proof}

\begin{lemma} 
\label{kkalgorithms: V(Dk)<=LIN(Dk)S+S(m(Dk)+k)/2}
	$V(D_k) \leq S \cdot LIN(D_k) \leq S \cdot OPT(D_k) \leq S \cdot LIN(D_k) + \frac{S \cdot (m(D_k) + k)}{2}$
\end{lemma}
\begin{proof}
	Let $\mathbf{c}$ and $\mathbf{n}$ be the vectors of dimension $m(D)$ such that for each natural $i\le m(D)$ the $i$th entry of $\mathbf{c}$ is the size $c[i]$ and the $i$th entry of $\mathbf{n}$ is the number $n[i]$ of items of size $c[i]$. 
 For each natural $j\le t$ let the size of the configuration $j$ be the sum of all sizes of the items of $a_j$, which is equal to the $j$th entry of the vector $\mathbf{c^T}A$. For instance, let $m(D)=2$, $n[1]=2$, $n[2]=2$, $c[1]=3$, $c[2]=4$, and $S=10$. Then, the number $t$ of feasible  configurations is $6$, and let the matrix $A$ be 
	$ \begin{bmatrix}
		0 & 0 & 1 & 2 &1 & 2 \\
		1 & 2 & 0 & 0 &1 & 1
	\end{bmatrix} $. Then $\mathbf{c^T}A=\begin{bmatrix}
		4 & 8 & 3 & 6 & 7 & 10
	\end{bmatrix}$. Note that the size of each configuration is at most the bin size $S$. So if $\mathbf{S}$ is a vector of dimension $t$ whose each entry equals $S$ then
	\begin{equation} \label{kk:lemma2:equation1}
		\mathbf{S} \geq \mathbf{c^T}A
	\end{equation}
	Let $\mathbf{x}$ be an optimal basic feasible solution of the fractional linear program $F_k$, which represents the $k$BP problem, for instance, $D_k$. 
        Each element $x[j]$ in $\mathbf{x}$ represents the (fractional) number of bins filled with configuration $j$.
	Then, $\mathbf{S} \cdot \mathbf{x}$ is the maximum possible sum of all items  
from
 all configurations in the solution vector $\mathbf{x}$.
 Then,
	\begin{align}
		S \cdot LIN(D_k) = S \cdot LIN(F_k) &= \mathbf{S} \cdot \mathbf{x}  \label{kk:lemma2:equation2}\\
		&\geq (\mathbf{c}^T A) \cdot \mathbf{x} && \text{by \ref{kk:lemma2:equation1}} \label{kk:lemma2:equation3} \\
		&\geq  \mathbf{c}^T(k\mathbf{n}) && \text{by \ref{eqn:dlvl:5}} 
	\label{kk:lemma2:equation4}
	\end{align}

Since for each natural $i\le m(D)$, the $i$th entry of $\mathbf{c^T}$ is the size $c[i]$ and the $i$th entry of $\mathbf{n}$ is the number of items of size $c[i]$, $\mathbf{c^T n}$ is the sum of sizes of all items of $D$, so by \ref{kk:lemma2:equation4} we have  $S \cdot LIN(D_k) = k \cdot V(D) = V(D_k)$.

\par Since for any natural $j\le t$ the variable $x_j$ in the problem for $LIN(D_k)$ is the fractional number of occurrences of the configuration $j$, we have $LIN(D_k) \leq OPT(D_k)$, and the equality holds iff the problem for $LIN(D_k)$ has an integer solution.

\par 

It is well-known that the linear program {\ref{eqn:dlvl:4}} - {\ref{eqn:dlvl:6}} has 
an optimal solution $\mathbf{x}$ which is a \emph{basic feasible solution},
where the number of
non-zero variables is at most the number $m(D)$ of constraints. 

For each natural $j\le t$ such that $x_j>0$ we have $x_j = \lfloor x_j \rfloor + r_j$, where $ \lfloor x_j \rfloor$ is the integer part of $x_j$, and $r_j$ is the fractional part of $x_j$.
For each natural $j \le t$ such that $r_j \in (0,1)$, we take configuration $j$. Let $D'$ be the resulting constructed instance. Clearly, all items in $D'$ will have a single copy.
Then,
\begin{equation} \label{kk:lemma2:equation6}
	V(D') \leq S \cdot LIN(D')  = S \cdot \sum_j r_j 
\end{equation}
For the residual part, we can construct two packings. 
First, for each non-zero $r_j$, let's take a bin of configuration $j$. We can then remove extra items.
Therefore,
\begin{equation} \label{eqn:kk:lemma2:1} OPT(D') \leq m(D_k) \text{ ( the instances } D \text{ and } D_k \text{ have the same configurations)}\end{equation}

Second, from Lemma \ref{kkalgorithms: opt-dk<=2V(Dk)/s+k} we have	 
\begin{equation} \label{eqn:kk:lemma2:2} OPT(D') \leq 2V(D')/S + k  \end{equation}
By \ref{eqn:kk:lemma2:1} and \ref{eqn:kk:lemma2:2},$$OPT(D')\le  \min \{m(D_k), 2V(D')/S + k\}\le\frac {m(D_k) +  2V(D')/S + k}2= V(D')/S + \frac{m(D_k)+k}{2}.$$.

Therefore,
\begin{align} \label{kklemm2e4}
	OPT(D_k) &\leq \text{the principal part} + OPT(D') \nonumber \\
	OPT(D_k) &\leq \text{the principal part} + V(D')/S + \frac{m(D_k)+k}{2} \nonumber \\
	OPT(D_k) &\leq LIN(D_k) + \frac{m(D_k)+k}{2} && \text{by \ref{kk:lemma2:equation2}}
\end{align}
\qed
\end{proof}

\begin{corollary} \label{kk:lemm2cor1}
	There is an algorithm \textbf{A} that solves the $k$BP by solving the corresponding fractional bin-packing problem $F_k$. It produces at most $\mathbf{1 \cdot x} + \frac{m(D_k)+k}{2}$ bins, where $\mathbf{x}$ is 
an optimal solution of $F_k$.
\end{corollary}
Note that the number of different item sizes in the instance is $m(D_k) = m(D)$.


\subsection{Solving the fractional linear program: } \label{kkalgorithms: solving-flp}
Solving the fractional linear program \hyperlink{kk:flp}{$F_k$}   involves a variable for each configuration. This results in a large number of variables. The fractional linear program $F_k$ has the following dual $D_F$.
\begin{align} \label{kk:dflp:appendix}
\max  \quad k \cdot \mathbf{n \cdot y} \\
\text{such that} \quad A^T \mathbf{y} &\leq \mathbf{1} \\
\mathbf{y} &\geq 0 
\end{align}
The above dual linear program can be solved to any given tolerance $h$ by using a variant of the ellipsoid method that uses an approximate separation oracle \cite{karmarkar-efficient-1982}.
This separation oracle accepts as input a vector $\mathbf{y}$ (each element $y[i]$ of $\mathbf{y}$ represents the price of item of size $c[i]$) and determines:
\begin{itemize}
\item whether $\mathbf{y}$ is feasible, or
\item if not, then return a constraint that is violated, i.e. a vector $\mathbf{a}$ such that $\mathbf{a} \cdot \mathbf{y} > 1$.
\end{itemize}
Let $\mathbf{c}$ be the $m(D)$-vector of the item sizes. The separation oracle solves the above problem by solving the following knapsack problem:
\begin{align} \label{kk:dflp:knapsack}
\max \quad \mathbf{a} \cdot \mathbf{y} \\
\text{such that} \quad \mathbf{a} \cdot \mathbf{c} &\leq S \\
\mathbf{a} &\geq 0 
\end{align}
Each entry $a[i]$ of $\mathbf{a}$ represents $a[i]$ pieces of size $c[i]$ 
, and hence $\mathbf{a}$ is an integer vector.   If $\mathbf{y}$ is feasible then the optimal value of the above knapsack problem is at most $1$. Otherwise, $\mathbf{a}$ correspond to a configuration that violates the constraint $\mathbf{a} \cdot \mathbf{y} \leq 1$.

Suppose we want a solution of $D_F$ 
within a specified tolerance $\delta$. Then, the above knapsack problem can be solved in polynomial time by rounding down each component of $\mathbf{y}$ to the closest multiple of  {$\frac{\delta}{n(D)}$}. The runtime of the above algorithm is  {$O\left(\frac{m(D) \cdot n(D)}{\delta}\right)$} \cite{karmarkar-efficient-1982}.  Alternatively, we can round down each component of $\mathbf{y}$ to the nearest multiple of {$\frac{\delta}{k \cdot n(D)}$}, and then the runtime of the above algorithm is {$O\left(\frac{k \cdot m(D) \cdot n(D)}{\delta}\right)$}.

The modified ellipsoid method uses the above approximate separation oracle as follows. Given the current ellipsoid center $\mathbf{y_c}$, let $\widetilde{\mathbf{y_c}}$ be the rounded down solution corresponding to $\mathbf{y_c}$. The modified ellipsoid method performs either an \textit{feasibility cut} (if $\widetilde{\mathbf{y_c}}$ is infeasible) or an \textit{optimality cut} (if $\widetilde{\mathbf{y_c}}$ is feasible). Since $\mathbf{\widetilde{y_c}}$ is the rounded down solution corresponding to $\mathbf{y_c}$, if $\mathbf{\widetilde{y_c}}$ is infeasible then $\mathbf{y_c}$ is not feasible too. In a feasibility cut, the modified ellipsoid method cuts from the ellipsoid all points such that $\mathbf{a} \cdot \mathbf{y} > 1$. If $\mathbf{\widetilde{y_c}}$ is feasible, then $\mathbf{y_c}$ may or may not be feasible.  
If the rounding-down is done to a multiple of {$\frac{\delta}{k \cdot n(D)}$}, then, by definition of the rounding, we have 
$$k \cdot \mathbf{n} \cdot \widetilde{\mathbf{y_c}} \geq k \cdot \mathbf{n} \cdot \mathbf{y_c} - k \cdot \mathbf{n} \cdot \mathbf{1} \cdot \frac{\delta}{k \cdot n} \nonumber  k \cdot \mathbf{n} \cdot \mathbf{y_c} - \delta.$$
In an optimality cut, the method preserves all the vectors whose value exceeds $\widetilde{\mathbf{y_c}}$ by more than $\delta$ \cite{karmarkar-efficient-1982} i.e. it removes all points that satisfy $k \cdot \mathbf{n} \cdot \mathbf{y} < k \cdot \mathbf{n} \cdot \widetilde{\mathbf{y_c}} + \delta $ .

The number of iterations in the modified ellipsoid algorithm is the same as that of the modified ellipsoid method in \cite{karmarkar-efficient-1982} that is, at most, {$Q = 4 \cdot {m(D)}^2 \cdot \ln{\left(\frac{m(D) \cdot n(D)}{\epsilon \cdot S \cdot \delta}\right)}$} iterations. The total execution time of the modified algorithm is  {$O\left(Q \cdot m(D) \cdot \left(Q \cdot m(D) + \frac{k \cdot n(D)}{\delta}\right)\right)$}. 

Let $t$ be the number of configurations in $k$BP, which is the same as the number of configurations in BP. There are at most $t$ constraints of the form $\mathbf{a} \cdot \mathbf{y} \leq 1$ during the ellipsoid method. It is well known that in any bounded linear program with {$m(D)$}  variables, at most {$m(D)$} constraints are sufficient to determine the optimal solution. 
We can use the same constraint elimination procedure as in \cite{karmarkar-efficient-1982} to come up with this critical set of {$m(D)$}  constraints. It results in having a reduced dual linear program with {$m(D)$}  variables and {$m(D)$}  constraints. Let us call this reduced dual linear program $R_{D_F}$. By taking the dual of $R_{D_F}$, we get a reduced primal linear program, say $R_F$. This reduced primal linear program will have only {$m(D)$}  variables corresponding to {$m(D)$}  constraints. Let us denote the optimal solutions of the dual linear program $D_F$ and the reduced dual linear program $R_{D_F}$ as $LIN(D_F)$ and $LIN(R_{D_F})$ respectively. Since $R_{D_F}$ is a relaxation of $D_F$ therefore, $LIN(D_F) \leq LIN(R_{D_F})$. From the modified GLS algorithm, we know that $LIN(R_{D_F}) \leq LIN(D_F) + \delta$ if rounding down is done in a multiple of  {$\frac{\delta}{k \cdot n(D)}$}. From the duality theorem of the linear program, both primal and dual have the same optimal solution. Therefore, the optimal solution of the reduced primal linear program, $R_F$, is at most $LIN(D_F) + \delta$ using the rounding down scheme in the modified GLS method.

The constraint elimination procedure will require to execute the modified GLS algorithm to run at most
{$(m(D)+1)\left(\left\lceil \frac{R}{m(D)+1} \right\rceil + (m(D)+1) \ln\left(\left\lceil \frac{R}{m(D)+1} \right\rceil\right) + 1\right)$} times, where  {$R \leq Q + 2 \cdot m(D)$} \cite{karmarkar-efficient-1982}. In the constraint elimination procedure, in the worst case, one out of every {$m(D)+1$} executions succeeds \cite{karmarkar-efficient-1982}. Therefore, the total accumulated error is  {$\delta \left(\left\lceil \frac{R}{m(D)+1} \right\rceil + (m(D)+1) \ln\left(\left\lceil \frac{R}{m(D)+1} \right\rceil\right) + 1\right)$}. If the tolerance $\delta$ is set to 
$$\frac{h}{\left\lceil \frac{R}{m(D)+1} \right\rceil + (m(D)+1) \ln\left(\left\lceil \frac{R}{m(D)+1} \right\rceil\right ) + 2}$$ then the optimal solution of the reduced primal linear program is at most $LIN(D_F) + h$ using the rounding down mechanism in the modified GLS method.

The total running time of the algorithm is
$$O\left(\left(Q  m(D) \left(Q  m(D) + \frac{k \cdot n(D)}{\delta}\right)\right)(m(D)+1)\left(\left\lceil \frac{R}{m(D)+1} \right\rceil + (m(D)+1) \ln\left\lceil \frac{R}{m(D)+1} \right\rceil + 1\right)\right) $$ $$\approx O\left(m(D)^8 \cdot  {\ln{m(D)}} \cdot {\ln^2\left(\frac{m(D) \cdot n(D)}{\epsilon \cdot S \cdot h}\right)} + \frac{m(D)^4 \cdot k \cdot n(D) \cdot \ln{m(D)}}{h} \ln{\frac{m(D) \cdot n(D)}{\epsilon \cdot S \cdot h}}\right).$$ {Since $m(D) = m(D_k)$ and $k \cdot n(D) = n(D_k)$}, in our analysis, we denote this running time as {$T(m(D_k),n(D_k))$}.	 

\subsection{Karmarkar-Karp Algorithm 1 extension to \texorpdfstring{$k$}{k}BP}
\label{kkalgorithms: algorithm1}
Recall that the algorithm 1 of Karmarkar-Karp algorithms uses the linear-grouping technique \ref{section:EAA:subsection:general}. 

\begin{algorithm}[H] 
\DontPrintSemicolon
\KwIn{A set $D$ of items, a number $\epsilon \in (0,1/2]$, and an integer $k$.}
\KwOut{A bin-packing of $D_k$.}
Let an item of $D$ is \emph{small} if it has size at  most $\max \{1/n, \epsilon\} \cdot S$ and \emph{large}, otherwise. Let $I$ and $J$ be multisets of small and large items of $D$, respectively.  Let $I_k$ be the $k$ copies of the instance $I$. \;
Perform the linear grouping on the instance $J$ by making groups of cardinality $g = \lceil n(J) \cdot {\epsilon}^2 \rceil $. Denote the resulting instances as ${U'}$  and  $U''$, as defined in the linear grouping technique in appendix 
\ref{appendix: dlvl-to-kbp} \;
Pack each item in ${U'_k}$ into a bin. There can be at most $g \cdot k$ bins, because each item of $D$ has $k$ copies in $D_k$. \;
Solve ${U''}_k$ using a fractional linear program with tolerance $h=1$. Denote the resulting solution as $\mathbf{x}$. \;
Construct the integer solution from $\mathbf{x}$ by the rounding method using no more than $1 \cdot \mathbf{x} + \frac{m(U'')+k}{2}$ bins (see Lemma \ref{kk:lemm2cor1}). \;
Reduce the item sizes as necessary and obtain a packing for $J_k$ along with the addition of bins as in step 3. \;
Obtain a packing for $D_k$ by adding the items in $I_k$ which are kept aside in step 1 by respecting constraint of $k$BP, creating new bin(s) when necessary. \;
\caption{Karmarkar-Karp Algorithm 1 extension to $k$BP}
\label{kk:algo1}
\end{algorithm}


\begin{proof}[of Theorem \ref{kkalgorithms:algorithm1:theorem:bin(Dk)<=(1+2kepsilon)OPT+additiveterms}] 
Since bin-packing is monotone.
\begin{equation} \label{eqn:kk:algorithm1:theorem1:1}
	OPT({U''}_k) \leq OPT(J_k) \leq OPT(D_k)  
\end{equation}
The size of each item in $J_k$ is at least $\frac{\epsilon}{2} \cdot S$. Therefore, from Lemma \ref{kkalgorithms: V(Dk)<=LIN(Dk)S+S(m(Dk)+k)/2},
\begin{align} \label{eqn:kk:algorithm1:theorem1:2}
	S\cdot OPT(J_k) &\geq V(J_k) \geq \frac{\epsilon}{2} \cdot S \cdot n(J_k) \nonumber \\
	OPT(J_k) &\geq \frac{\epsilon}{2} \cdot n(J_k).
\end{align}
From algorithm step 2 and from equations \ref{eqn:kk:algorithm1:theorem1:1} and \ref{eqn:kk:algorithm1:theorem1:2},
\[ g = \lceil n(J_k) \epsilon^2 \rceil \leq 2 \cdot \epsilon \cdot OPT(J_k) + 1 \leq 2 \cdot \epsilon \cdot OPT(D_k) + 1  \]
\[ m({U''}_k) = n({U''}_k)/g = n({U''}_k)/{ \lceil n(J_k)\epsilon^2 \rceil } \leq n(J_k)/ { \lceil n(J_k)\epsilon^2 \rceil } \leq 1/\epsilon^2.  \]
We have tolerance $h=1$, so
\[ \mathbf{1 \cdot x} \leq LIN({U''}_k) + 1 \leq OPT({U''}_k) + 1 \leq OPT(D_k) +1 \]
Step $6$ will result in the number of bins which is at most
\[ \mathbf{1 \cdot x} + \frac{m({U''}_k) + k}{2} + g \cdot k  \leq\]
\[   OPT(D_k) +1 + \frac{k + \frac{1}{\epsilon^2}}{2} + k \cdot (2 \cdot \epsilon \cdot OPT(D_k) + 1) \leq\]
\[ \ (1 + 2 \cdot k \cdot \epsilon)OPT(D_k) + \frac{1}{2 \cdot \epsilon^2} + (2k+1). \]

By Lemma \ref{general:lemma1}, the number $bins(D_k)$ of bins required at step $7$ of the algorithm is at most 
$$\max \left\{ (1 + 2 \cdot k \cdot \epsilon)OPT(D_k) + \frac{1}{2 \cdot \epsilon^2} + (2k+1), (1 + 2\ \cdot \epsilon)OPT(D_k) + k \right\}=$$ $$(1 + 2 \cdot k \cdot \epsilon)OPT(D_k) + \frac{1}{2 \cdot \epsilon^2} + (2k+1). \notag \nonumber $$
\qed
\end{proof}

\paragraph{Running time of algorithm \ref{kk:algo1}.}  \label{kk:algorithm1:running-time}
The time taken by fractional linear program is at most $T(m(U''_k),n(U''_k) ) \leq T\left(\frac{1}{\epsilon^2}, n(D_k)\right)$. The rest of the steps in Algorithm 1 will take at most $O(n(D_k) \log {n(D_k)})$ time. Therefore, the execution time of above Algorithm 1 is at most $O\left(n(D_k) \log {n(D_k)} + T(\frac{1}{\epsilon^2}, n(D_k)\right)$.


\subsection{Karmarkar-Karp Algorithm 2 extension to \texorpdfstring{$k$}{k}BP} \label{kkalgorithms:algorithm2}

\paragraph{Alternative Geometric Grouping.} Let $J$ be an instance and $g>1$ be an integer parameter. Sort the items in $J$ in a non-increasing order of their size. Now, we partition $J$ into groups $G_1,\dots G_r$ 
each containing necessary number of items so that the size of each but the last 
group (the sum of item sizes in that group) is at least $g \cdot S$.

For natural $i\le r$ let $l_j$ be the  cardinality of the group $G_i$, that is the number of items in the group.
Let 
$l_r$ be the cardinality of $G_r$. Note that the size of $G_r$ may be less than $g \cdot S$.

Since each item size is less than $S$, the number of items in each but the last
group is at least $g+1$. Let $\epsilon \cdot S$ be the smallest item size in $J$. Then, the number of items in each but the last 
group is at most $\frac{g}{\epsilon}$. 
Therefore,
$
g+1 \leq l_1 \leq \ldots \leq l_r \leq \frac{g}{\epsilon}
$. 
For each natural $i<r$ let ${G'}_{i+1}$ be the group obtained by increasing the size of each item in ${G}_{i+1}$ to the maximum item size in that group except for the smallest $l_{i+1} - l_{i}$ items in that group. Let $U' = \bigcup_{i=2}^r \Delta {G_i} \cup G_1$ and $U'' = \bigcup_{i=2}^r G_i^{'}$, where $\Delta G_i$ is the multiset of smallest $l_i - l_{i-1}$ items in $G_i$ for each natural $i$ with $2\le i\le r$ .

Since bin-packing is monotone, $OPT(J) \leq OPT(U') + OPT(U'')$. By Lemma \ref{kkalgorithms: opt-dk<=2V(Dk)/s+k}, $OPT(U') \leq 2 \cdot V(U')/S + 1$ (note that $k=1$). It holds that $V(U') \leq S \cdot g \cdot (1 + \ln {\frac{1}{\epsilon \cdot S}}) + S$ \cite[page 315]{karmarkar-efficient-1982}. Therefore,
\begin{equation} \label{kk:equation4}
OPT(J) \leq OPT(U'') + 2\cdot g \cdot \left(2 + \ln {\frac{1}{\epsilon \cdot S}} \right)
\end{equation}
Also, $V(U'') \geq {\sum_{i=2}^r} g \cdot S \cdot \frac{l_{i-1}}{l_i} \geq g \cdot S \cdot \left( (r-1) - \ln \frac{1}{\epsilon \cdot S} \right)$. 
Therefore, by \cite{karmarkar-efficient-1982},  
\begin{equation} \label{kk:equation5}
m(U'') = r-1 \leq \frac {V(U'')}{g \cdot S} + \ln {\frac{1}{\epsilon \cdot S}}.
\end{equation} 

\begin{algorithm}[H]
\DontPrintSemicolon
\KwIn{A set $D$ of items, $\epsilon \in (0,1/2]$, and integers $g > 1$ and $k$.}
\KwOut{A bin-packing of $D_k$.}
Let $J$ be the instance obtained after removing all items of size at most $ \epsilon \cdot S$ from $D$.  \;
\While{$V(J) > S \cdot (1 + \frac{g}{g-1} \cdot \ln{\frac{1}{\epsilon}})$}{
	Do the alternative geometric grouping with the parameter $g$. Let the resulting instances be $U'$ and $U''$.  \;
	Solve ${U''}_k$ using fractional bin-packing with the tolerance $h=1$. Denote the resuting basic feasible solution as $\textbf{x}$. \;
	Construct an integer solution from $\textbf{x}$. Create $\lfloor x_i \rfloor$ bins for each $x_i$ in $\textbf{x}$. Remove the packed items from $J$. \;
	Pack ${U'}_k$ in  at most $2 \cdot k \cdot g \cdot \left[ 2 + \ln \frac{1}{\epsilon} \right]$ bins. \;
}
When $V(J) \leq S \cdot (1 + \frac{g}{g-1} \cdot \ln{\frac{1}{\epsilon}})$, pack the remaining items greedily in at most $(2 + \frac{2 \cdot g}{g-1} \cdot \ln {\frac{1}{\epsilon}})$ bins and create $k$ copies of this packing. \;
Greedily pack the $k$ copies of the items of size at most $\epsilon \cdot S$ respecting $k$BP constraint to get a solution for $D_k$. \;
\caption{Karmarkar-Karp Algorithm 2 extension to $k$BP}
\label{kk:algorithm2}
\end{algorithm}	


\begin{proof}[Proof of Theorem \ref{kkalgorithm2:theorem:bin(Dk)<=opt+O(klog2opt)}]
Let ${U''}_{k,i}$ 
be the instance of $U''_k$ 
at the beginning of the $i$th iteration of the loop in step $2$. Then the fractional part of $\mathbf{x}$ contains at most $m({U''}_{k,i})$ 
non-zero variables $x_j$ for $j\le t$. Therefore, from 
\eqref{kk:equation5} we have
\begin{equation} \label{kk:algorithm2: analysis:equation1}
	m({U''}_{k,i})  = m({U''}_{1,i}) \leq \frac{V({U''}_{1,i})}{g\cdot S} + \ln{\frac{1}{\epsilon \cdot S}}
\end{equation}
\emph{The number of iterations in step $6$ of algorithm.}
As shown in \cite{karmarkar-efficient-1982} the number of iterations is at most $\frac{\ln {V({D})}}{\ln {g}}+ 1$.

\emph{Number of bins produced using Fractional Bin-Packing solution in steps $3-5$.}
Let $J_i$ denote the item set 
$J$ 
at the beginning of $i$th iteration. On applying alternative geometric grouping with the parameter $g$, we get 
the instances ${U'}_{1,i}$ and ${U''}_{1,i}$ . Let 
${U'}_{k,i}$ and ${U''}_{k,i}$ be the corresponding $k$BP instances.
Let $B_i$ be the number of bins obtained in solving ${U''}_{k,i}$ by applying fractional bin-packing at the iteration $i$ with the tolerance $h=1$. Then
\[
B_i \leq LIN({U''}_{k,i}) + 1.
\]
Summing across all iterations
\begin{align} \label{kk:algorithm2: analysis:equation2}
	\sum B_i &\leq LIN({U''}_{k}) + \frac{\ln {V(D) }}{\ln {g}}+ 1 \nonumber \\
	&\leq OPT({U''}_{k}) + \frac{\ln {V({D}) }}{\ln {g}}+ 1
\end{align}
\emph{The total number of bins used.}
Let $bins(D_k)$ denote the total number of used bins
Then $bins(D_k)$ is the sum of  the bins used 
in steps $3-5$, in step $6$, and in step $8$. So 
\begin{equation}
	\begin{split}
		bins(D_k) \leq OPT(D_k) + \left(\frac{\ln{V(D)}}{\ln{g}} + 1\right)\left(1 + 4 \cdot k \cdot g + 2 \cdot k \cdot g \ln{\frac{1}{\epsilon}}\right) + \\ 
		k\cdot \left(2 + \frac{2 \cdot g}{g-1} \ln{\frac{1}{\epsilon}}\right).	
	\end{split}
\end{equation}
By Lemma \ref{general:lemma1}, 
\begin{equation}
	\begin{split}
		bins(D_k) \leq \max 
		\bigg\{
			OPT(D_k) + 
			\left(
				\frac{\ln{V(D)}}{\ln{g}} + 1
			\right)
			\left(
				1 + 4 \cdot k \cdot g + 2 \cdot k \cdot g \ln{\frac{1}{\epsilon}}
			\right) + 
			\\ 
			k\cdot
			\left(
				2 + \frac{2 \cdot g}{g-1} \ln{\frac{1}{\epsilon}}
			\right), 
			(1+2 \cdot \epsilon)\cdot OPT(D_k) + k  
		\bigg\}
	\end{split}
\end{equation}
Choosing $g=2$ and $\epsilon=\frac{1}{V(D)}$ will result in 
\begin{align}
	bins(D_k) &\leq OPT(D_k) + O(k \cdot \log^2 {OPT(D)})
\end{align}
number of bins.
\qed
\end{proof}

\paragraph{Running time of algorithm \ref{kk:algorithm2}}
By \ref{kk:equation5}, at each iteration of the while loop, $m({U''}_{1,i})$ decreases by a factor of $g \cdot S$. During the $i$th iteration of the while loop, the time taken by fractional bin-packing is at most $T(m({U''}_{k,i}), n({U''}_{k,i}))$. The remaining steps will take at most $O(n(D_k) \cdot \log {n(D_k)})$ time. Therefore, the total time taken by the algorithm is $O\left(T\left(\frac{n(D) \cdot S}{g \cdot S} + \ln {\frac{1}{\epsilon}},n(D_k)\right) + n(D_k) \cdot \log{n(D_k)}\right)$. For the choice $g=2$ and $\epsilon=\frac{1}{V(D)}$ the running time will be $O\left(T\left(\frac{n(D)}{2} + \ln {V(D)},n(D_k)\right) + n(D_k) \cdot \log{n(D_k)}\right) \in O\left(T\left(\frac{n(D)}{2},n(D_k)\right) + n(D_k) \cdot \log{n(D_k)}\right)$.



\section{Experiment: approximation ratio of \texorpdfstring{$FFk$}{FFk} and \texorpdfstring{$FFDk$}{FFDk} algorithms}
\label{subsection:ffk and ffdk bp algorithms experiment}
In this section we describe the second experiment, for checking the approximation ratio of $FFk$ and $FFDk$ with respect to the optimal bin-packing.

\subsection{Data generation}
Since computing the optimal bin-packing of a given instance is NP-hard, we constructed instances whose optimal bin-packing is known. Given the bin capacity $S$, we use Algorithm \ref{appendix:experimentalResults:generate_ll}to generate a list of 
item sizes that sum up to the bin capacity $S$. Algorithm \ref{appendix:experimentalResults:groupItemsGeneration} uses Algorithm \ref{appendix:experimentalResults:generate_ll} to generate an instance whose 
item sizes sum up to $OPT$ times $S$, and by the construction, we know that the optimal bin-packing for the instance $D$ has exactly $OPT$ bins,
and moreover, the optimal bin-packing for $D_k$ has exactly $k\cdot OPT$ bins (all of them full).

\begin{algorithm}[H]
	\DontPrintSemicolon
	\KwIn{A bin capacity ${S}$}
	\KwOut{A list $L$ of item sizes whose sum is $S$.}
	\SetKw{KwRet}{return}
	
	Let $L$ be the list of the randomly generated item sizes. Initially $L$ is empty. \;
	Let $s_L$ be the sum of all the item sizes in the list. Initially $s_L=0$. \;
	
	\While{True}
	{generate a random number $r\in (1,S)$. \; 
		\eIf{$s_L + r > S$}
		{Append $S - s_L$ to $L$. \;
			break \;}
		{Append $r$ to $L$, and update $s_L$. \;}
	}
	return $L$
	
	\caption{generate\_items}
	\label{appendix:experimentalResults:generate_ll}
\end{algorithm}

Our generated dataset contains multiple files. Each file in the dataset contains a fixed number ${IC}$ of instances. Each instance has several item sizes. These items have an optimal packing into ${OPT}$ bins. 
Each bin has capacity $S$.
For example let $[1,2,3,7,6,1']$ be an instance in some file where for each instance in the file ${OPT}=2$. Bin capacity is $S=10$. 
As we can see, there are two groups $\{7,2,1\}, \{6,3,1'\}$ which sum to 10 and hence an optimal packing requires two bins to pack the items in this instance. 

\begin{algorithm}
	\DontPrintSemicolon
	\KwIn{A bin capacity ${S}$ and the optimal number ${OPT}$ of bins to pack items in this instance.}
	\KwOut{A list $D$ of item sizes which sum to ${OPT}\cdot{S}$, such that the optimal number of bins is indeed $OPT(D)=OPT$.}
	
	Let $D$ be the instance of item sizes. Initially $D$ is empty. \;
	Let $s_D$ be the sum of all item sizes in $D$. Initially $s_D=0$. \;
	\While{True}
	{$L =$ generate\_items($S$) \;
		$s_D = s_D + s_L$. \;
		\eIf{$s_D / S == {OPT}$}
		{Append item sizes in $L$ to $D$. \;
			Shuffle item sizes in $D$. \;
			break \;}
		{Append item sizes in $L$ to $D$. \;}
	}
	return $D$ \;
	
	\caption{generate\_instance\_items}
	\label{appendix:experimentalResults:groupItemsGeneration}
\end{algorithm}	




\subsection{Results}
\label{section:results}
\label{subsection:ffkffdkresults}

We ran our experiment for different values of $k > 1$, and for different integer values of $OPT$, $2\le OPT \le 9$.
We computed the conjectured upper bound on the number of bins required to pack items using $FFk$ (for $k>1$) and $FFDk$ as $1.375 \cdot OPT(D_k)$ and $\frac{11 \cdot OPT(D_k) + 6}{9}$, respectively. For each file, we report the maximum number of bins required to pack items. 
We can see at Figures \ref{figure:ffk-ffdk:k2k3}, and \ref{figure:ffk-ffdk:k4k5}
that the bins used by $FFk$ (for $k>1$) and $FFDk$ are at most $1.375 \cdot OPT(D_k)$ and  $\frac{11 \cdot OPT(D_k) + 6}{9}$, respectively, which supports our conjectures
in subsections \ref{section:approx-algo:subsection:approx_algo:ffk} and \ref{section:approx-algo:subsection:approx_algo:ffdk}.


\begin{figure}[p]
	\begin{subfigure}[0.8\textheight]{\linewidth}
		\centering
		\includegraphics[height = 0.45\textheight, width=\textwidth]{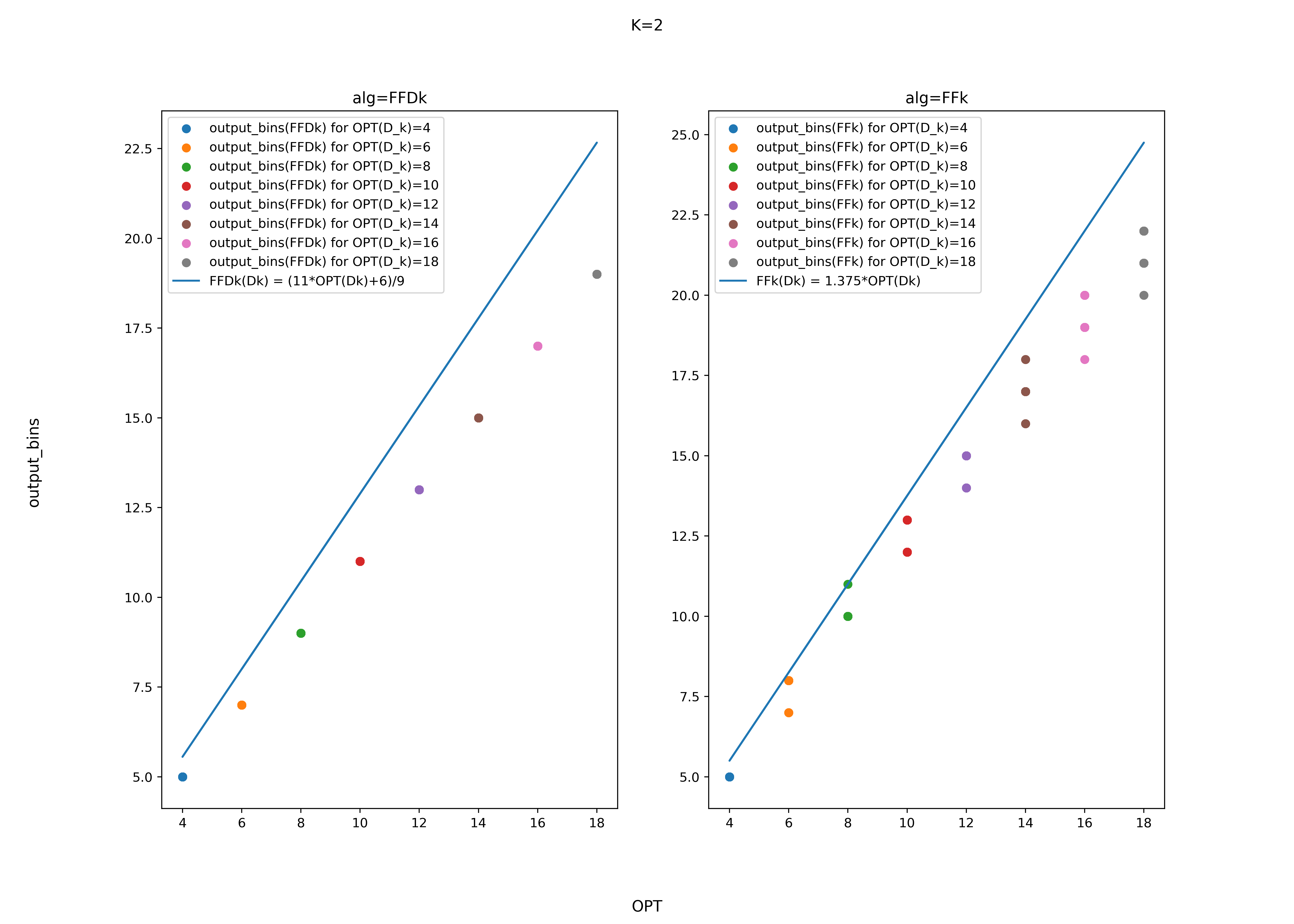}
	\end{subfigure}
	
	\begin{subfigure}[0.8\textheight]{\linewidth}
		\centering
		\includegraphics[height = 0.45\textheight, width=\textwidth]{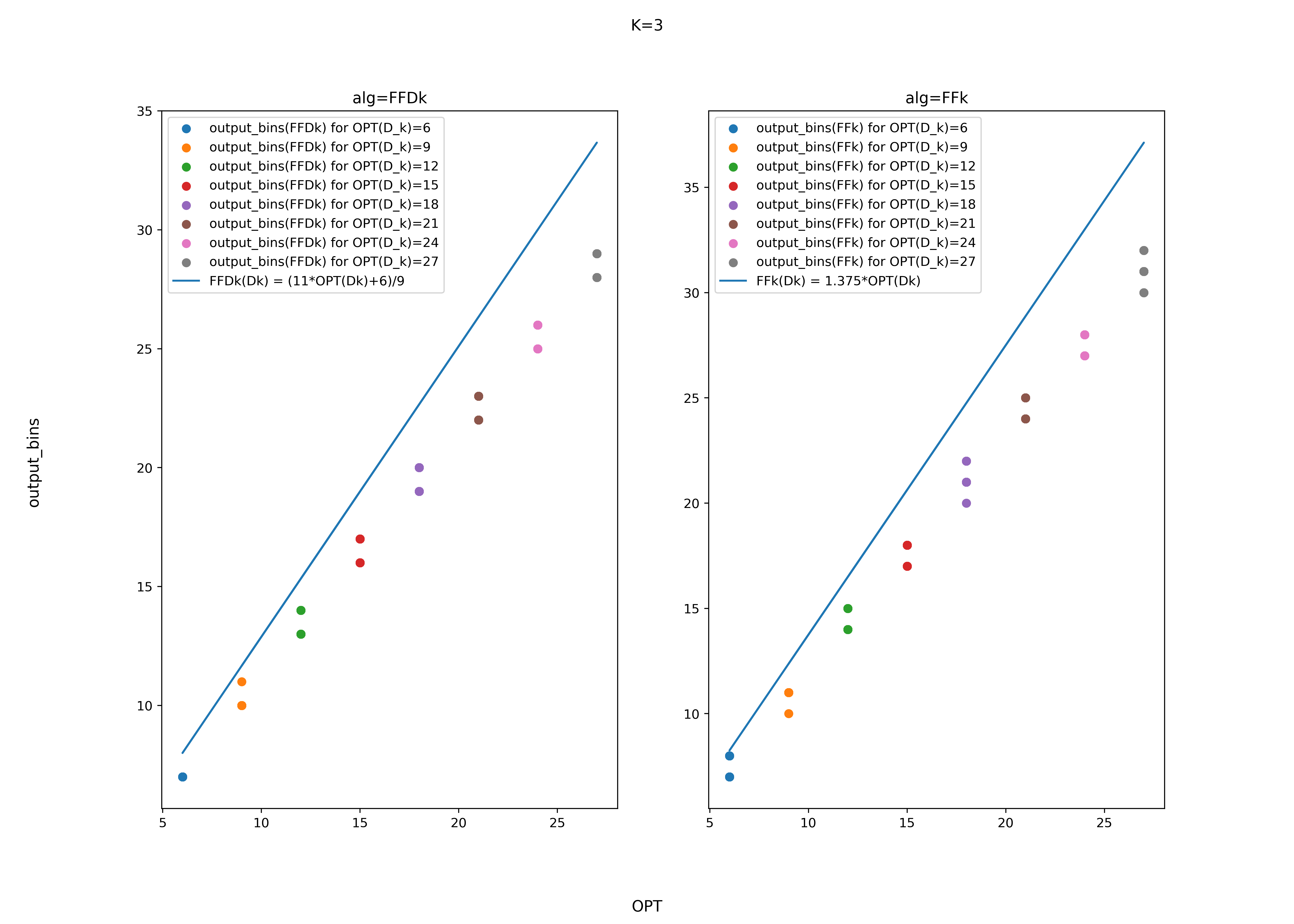}
	\end{subfigure}		
        \caption{The optimal numbers $OPT(D_k)$ of bins and numbers of bins used by $FFk$ and $FFDk$ algorithms bins are shown at the $x$- and $y$-axis, respectively. The data points in the legend show the upper bound that can be hypothesized for $OPT(D_k) = k \cdot OPT(D)$. Each subplot's data points display the number of different output bins corresponding to each conjectured upper bound. Note that the output bins convey the conjectured upper bounds.}
	\label{figure:ffk-ffdk:k2k3}		
\end{figure}


\begin{figure}[p]
	\begin{subfigure}[0.8\textheight]{\linewidth}
		\centering
		\includegraphics[height = 0.45\textheight, width=\textwidth]{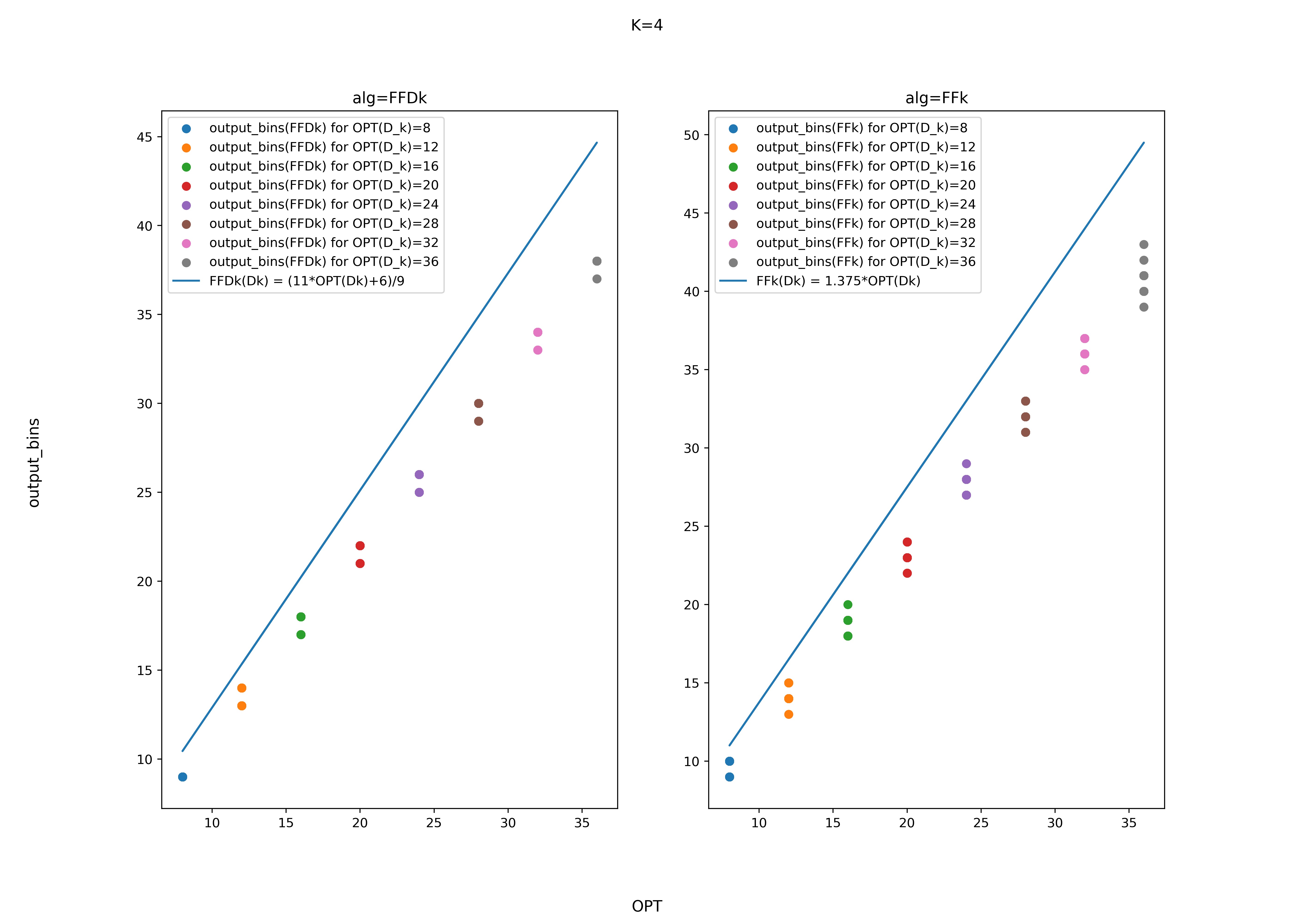}
	\end{subfigure}
	
	\begin{subfigure}[0.8\textheight]{\linewidth}
		\centering
		\includegraphics[height = 0.45\textheight, width=\textwidth]{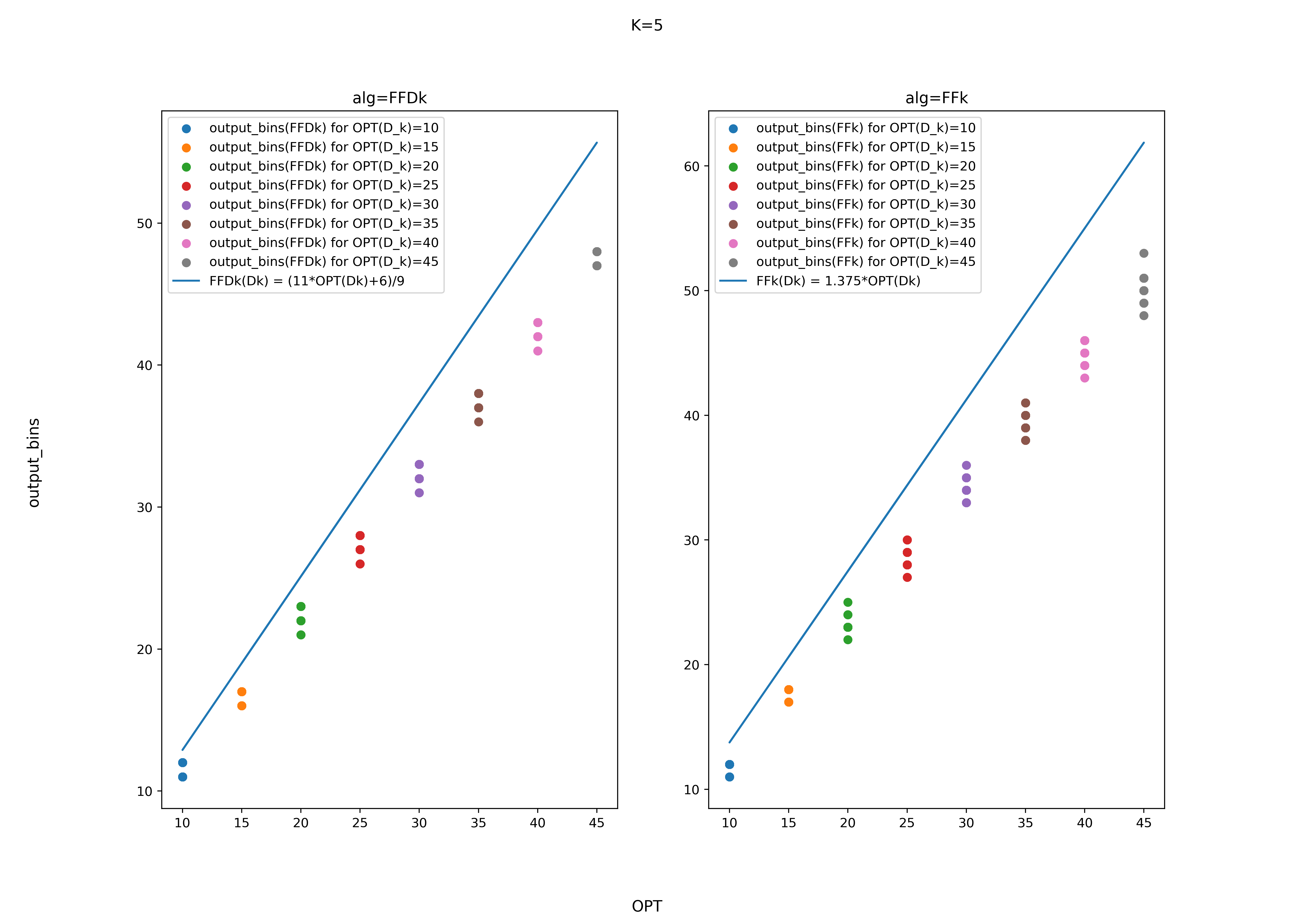}
	\end{subfigure}		
        \caption{The optimal numbers $OPT(D_k)$ of bins and numbers of bins used by $FFk$ and $FFDk$ algorithms bins are shown at the $x$- and $y$-axis, respectively. The data points in the legend show the upper bound that can be hypothesized for $OPT(D_k) = k \cdot OPT(D)$. Each subplot's data points display the number of different output bins corresponding to each conjectured upper bound. Note that the output bins convey the conjectured upper bounds.}
	\label{figure:ffk-ffdk:k4k5}		
\end{figure}

\section{More Electricity Distribution Results}	\label{APPENDIX: electricity-distribution-results}
In this section, we discuss the variation in hours of connection, comfort, and electricity supplied with different values of $k$ and with varying levels of uncertainty (standard deviation) in the agent's demand. 

In Figure \ref{appendix:more-results:k-vs-comf-util-with-delta}, we can observe that the sum of comfort the algorithm delivers to agents increase with an increasing value of $k$. However, that increase appears to saturate as $k$ increases. Similar phenomena occur for other metrics.

Figure \ref{appendix:more-results:k-vs-comf-egal-with-delta} shows that the minimum comfort the algorithm delivers to agents individually increases with an increasing value of $k$. However, that increase appears to saturate as $k$ increases. Similar phenomena occur for other metrics.


\begin{figure}[h]
\centering\includegraphics[width=\linewidth, height=0.45\textheight]{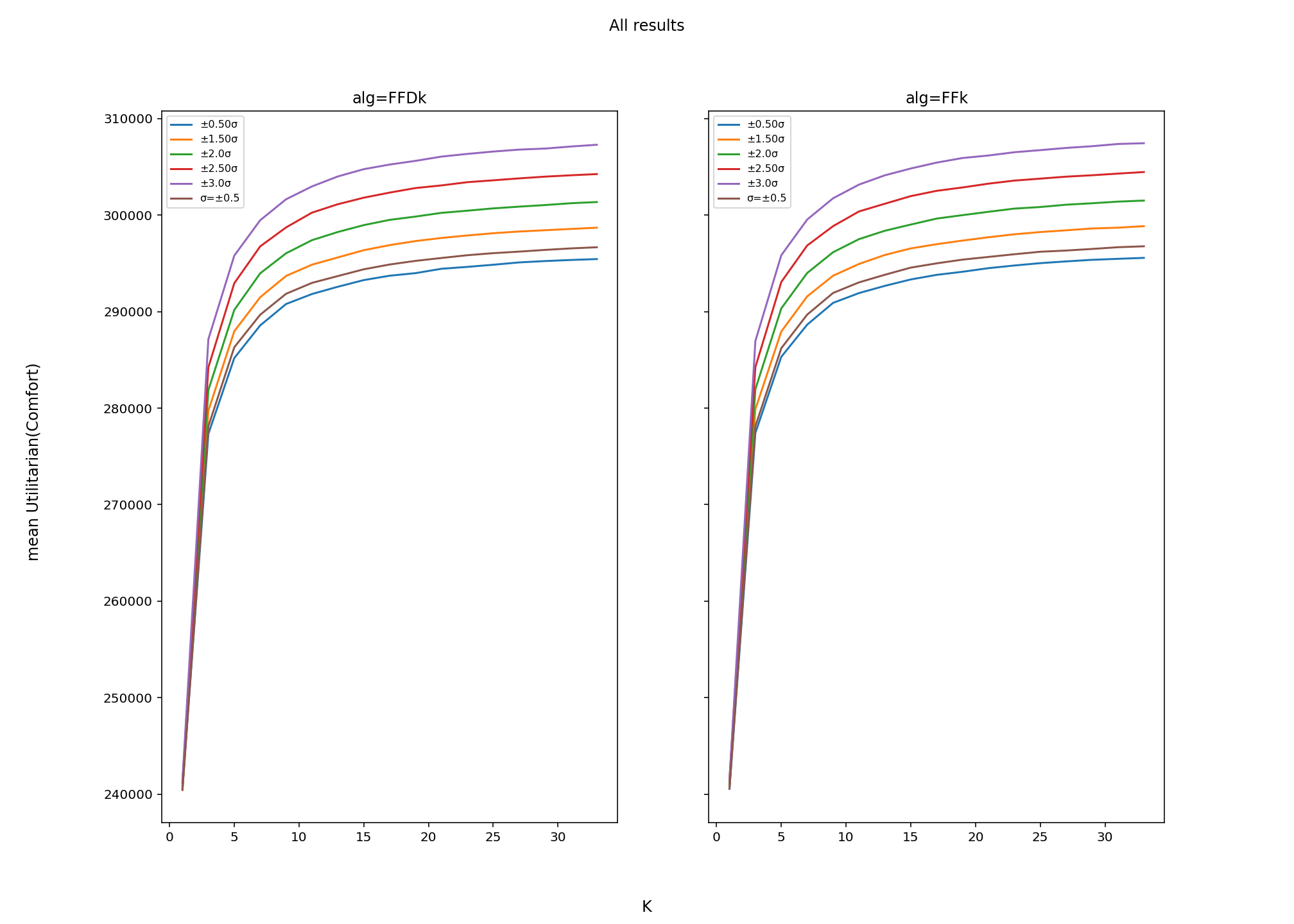}

\caption{Values of $k$ and mean utilitarian value of comfort are shown at the $x$- and $y$-axis, respectively. Figure shows the increase in sum of comfort delivered by the algorithm to agents with increasing value of $k$ for varying levels of uncertainty $\sigma$. 
\label{appendix:more-results:k-vs-comf-util-with-delta}
}


\centering\includegraphics[width=\linewidth, height=0.45\textheight]{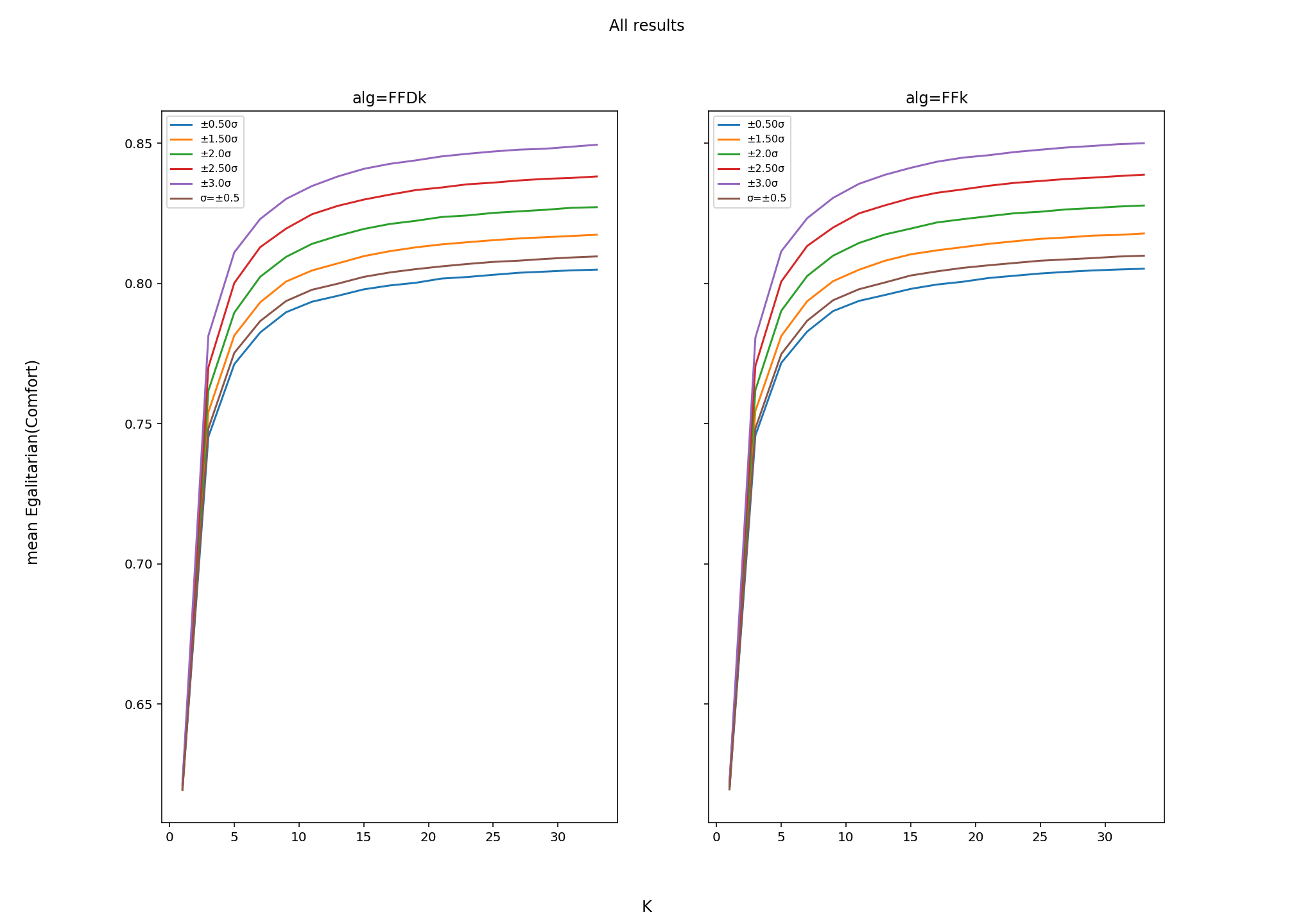}
\caption{Values of $k$ and minimum egalitarian value of comfort are shown at the $x$- and $y$-axis, respectively. Figure shows the increase in minimum comfort delivered to agents individually with increasing value of $k$ for varying levels of uncertainty $\sigma$. 
\label{appendix:more-results:k-vs-comf-egal-with-delta}
}
\end{figure}

In Figure \ref{appendix:more-results:k-vs-comf-mud-with-delta}, the maximum utility difference decreases with an increasing value of $k$. 
Similar phenomena occurs for electricity-supplied social welfare metrics. However, in the case of hours of connection, the maximum utility difference is 0 because the algorithm is executed for each hour.
\begin{figure}[h]
\centering\includegraphics[width=\linewidth, height=0.45\textheight]{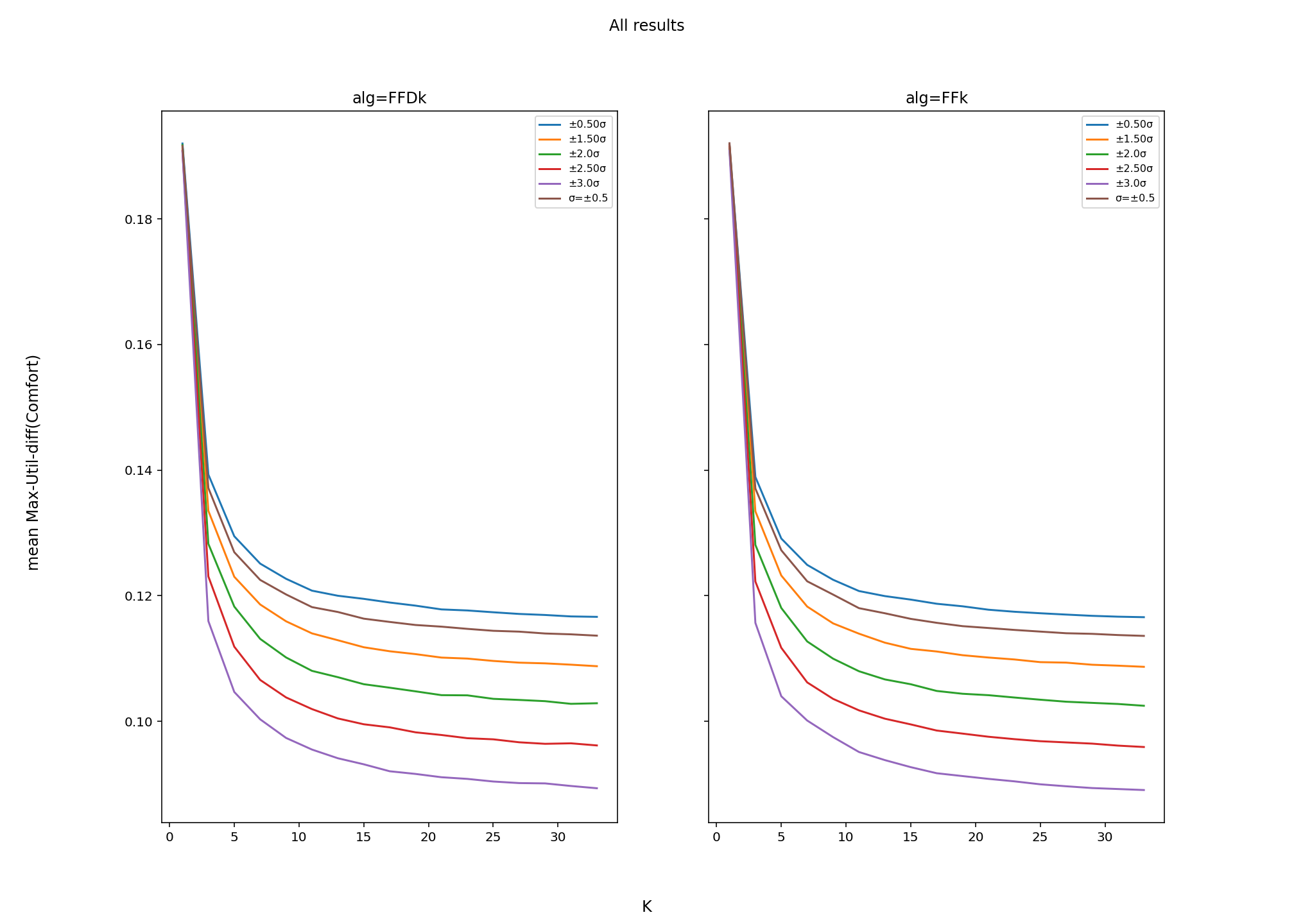}
\caption{Values of $k$ and maximum utility difference of comfort are shown at the $x$- and $y$-axis, respectively. Figure shows the decrease in maximum utility difference with increasing value of $k$ for varying levels of uncertainty $\sigma$. }
\label{appendix:more-results:k-vs-comf-mud-with-delta}
\end{figure}

\begin{credits}
\subsubsection{\ackname} We thank the reviewers of SAGT 2024 for their constructive comments. This research is partly funded by the Israel Science Foundation grant no. 712/20.
\end{credits}

\end{document}